\documentclass[11pt]{article}

\usepackage[english]{babel}
\usepackage{xcolor}
\usepackage{fullpage}
\usepackage{verbatim}
\usepackage[utf8]{inputenc}
\usepackage{amsmath}
\usepackage{amsfonts}
\usepackage{amsthm}
\usepackage{amssymb}
\usepackage{graphicx}
\usepackage[colorinlistoftodos]{todonotes}
\usepackage{caption}
\usepackage{subcaption}
\usepackage{hyperref}
\usepackage{authblk}
\usepackage[linesnumbered,lined,commentsnumbered]{algorithm2e}
\usepackage{xspace}

\usepackage{multicol}
\newenvironment{itemizeTwoCol}{%
\begin{itemize}
\begin{multicols}{2}
}{%
\end{multicols}
\end{itemize}
}

\newtheorem{theorem}{Theorem}
\newtheorem{lemma}[theorem]{Lemma}

\newtheorem{definition}[theorem]{Definition}
\newtheorem{corollary}[theorem]{Corollary}
\newtheorem{claim}[theorem]{Claim}
\newtheorem{remark}[theorem]{Remark}
\newcommand{\E}{\mathbb{E}}
\newcommand{\Prob}{\mathbf{Pr}}
\newcommand{\NP}{\textrm{NP}}
\renewcommand{\P}{\textrm{P}}
\newcommand{\poly}{\texttt{poly}}
\DeclareMathOperator {\Harm}{Harm}
\DeclareMathOperator {\Path}{Path}

\newcommand{\HLp}{HL$_p$\xspace}
\newcommand{\HL}[1]{HL$_#1$\xspace}

\begin{document}

\title{Algorithmic and Hardness Results for the Hub Labeling Problem}
\author[1]{Haris Angelidakis}
\author[1]{Yury Makarychev}
\author[2]{Vsevolod Oparin}
\affil[1]{\small Toyota Technological Institute at Chicago}
\affil[2]{\small Saint Petersburg Academic University of the Russian Academy of Sciences}

\renewcommand\Authands{ and }
\date{}

\maketitle

\begin{abstract}
There has been significant success in designing highly efficient algorithms for distance and shortest-path queries in recent years;
many of the state-of-the-art algorithms use the hub labeling framework. In this paper,
we study the approximability of the Hub Labeling problem. We prove a hardness of $\Omega(\log n)$ for Hub Labeling, matching known approximation guarantees. 
The hardness result applies to graphs that have multiple shortest paths between some pairs of vertices.
No hardness of approximation results were known previously.

Then, we focus on graphs that have a unique shortest path between each pair of vertices. This is a very natural family of graphs,
and much research on the Hub Labeling problem has studied such graphs.
We give an $O(\log D)$ approximation algorithm for graphs of diameter $D$ with unique shortest paths.
In particular, we get an $O(\log \log n)$ approximation for graphs of polylogarithmic diameter, while
previously known algorithms gave an $O(\log n)$ approximation.
Finally, we present a polynomial-time approximation scheme (PTAS) and
quasi-polynomial time algorithms for Hub Labeling on trees; additionally, we analyze a simple combinatorial heuristic for Hub Labeling on trees,
proposed by Peleg in 2000. We show that this heuristic gives an approximation factor of 2.
\end{abstract}


\section{Introduction}
There has been significant success in designing highly efficient algorithms for distance and shortest-path queries in recent years;
many of the state-of-the-art algorithms use the hub labeling framework.
In this paper, we present approximation algorithms as well as prove hardness results for the Hub Labeling problem.
The Hub Labeling problem was introduced by Cohen et al.%
\footnote{See also prior papers on bit labeling schemes~\cite{DBLP:journals/tit/Breuer66, Breuer1967583,DBLP:journals/siamdm/KannanNR92,Graham1972,DBLP:journals/combinatorica/Winkler83,DBLP:journals/jgt/Peleg00,DBLP:journals/jal/GavoillePPR04}.
} in 2003~\cite{DBLP:journals/siamcomp/CohenHKZ03}.
\begin{definition}[Hub Labeling]
Consider an undirected graph $G=(V,E)$ with edge lengths $l(e) > 0$. Suppose that we are given a set system $\{H_u\}_{u\in V}$ with
one set $H_u\subset V$ for every vertex $u$.
We say that  $\{H_u\}_{u\in V}$ is a hub labeling if it satisfies the following covering property:
for every pair of vertices $(u, v)$  ($u$ and $v$ are not necessarily distinct), there is a vertex in $H_u \cap H_v$ (a common ``hub'' for $u$ and $v$) that
lies on a shortest path between $u$ and $v$. We call vertices in sets $H_u$ hubs: a vertex $v\in H_u$ is a hub for $u$.
\end{definition}
\noindent In the Hub Labeling problem (HL), our goal is to find a hub labeling with a small number of hubs; specifically, we want to minimize the $\ell_p$-cost of a hub labeling.
\begin{definition}
The $\ell_p$-cost of a hub labeling $\{H_u\}_{u\in V}$ equals $(\sum_{u \in V} |H_u|^p)^{1/p}$ for $p\in[1,\infty)$;
the $\ell_\infty$-cost is $\max_{u \in V} |H_u|$. The hub labeling problem with the $\ell_p$-cost, which we denote by \HLp,
asks to find a hub labeling with the minimum possible $\ell_p$-cost.
\end{definition}
\noindent \textit{\textbf{Note:} When we talk about HL and do not specify the cost function explicitly, we mean \HL1;
we sometimes refer to the $\ell_1$-cost of $\{H_u\}_{u\in V}$ simply as the cost of $\{H_u\}_{u\in V}$.}

We are interested in the Hub Labeling problem because of its connection to the shortest-path problem.
Nowadays hundreds of millions of people worldwide use web mapping services and GPS devices
to get driving directions. That creates a huge demand for fast algorithms for computing shortest paths (algorithms that are even faster
than the classic Dijkstra's algorithm). Hub labelings provide a highly efficient way for computing shortest paths
(see also the paper of Bast et al. \cite{DBLP:journals/corr/BastDGMPSWW15} for a review and discussion of various methods
for computing shortest paths that are used in practice).

Let us briefly explain the connection between the Hub Labeling and shortest-path problems.
Consider a graph $G=(V,E)$ with edge lengths $l(e)$. Let $d(u,v)$ be the shortest path metric on $G$.
Suppose that we have a hub labeling $\{H_u\}_u$. During the preprocessing step, we compute and store the distance $d(u,w)$ between every
vertex $u$ and each hub $w\in H_u$ of $u$. Observe that now we can very quickly answer a distance query: to find $d(u,v)$ we compute
$\min_{w\in H_u \cap H_v} (d(u,w) + d(v,w))$. By the triangle inequality, $d(u,v) \leq \min_{w\in H_u \cap H_v} (d(u,w) + d(v,w))$,
and the covering property guarantees that there is a hub $w\in H_u \cap H_v$ on a shortest path between $u$ and $v$; so
$d(u,v) = \min_{w\in H_u \cap H_v} (d(u,w) + d(v,w))$. We can compute $\min_{w\in H_u \cap H_v} \left(d(u,w) + d(v,w) \right)$
and answer the distance query in time  $O(\max(|H_u|, |H_v|))$.
We need to keep a lookup table of size $O(\sum_{u\in V} |H_u|)$ to store the distances between the vertices and their hubs.
So, if, say, all hub sets $H_u$ are of polylogarithmic size, the algorithm answers a distance query in polylogarithmic time
and requires $n \mathop{\mathrm{polylog}} n$ space.
The outlined approach can be used not only for computing distances but also shortest paths between vertices.
It is clear from this discussion that it is important to have a hub labeling
of small size, since both the running time and storage space depend on the number of hubs.

Recently, there has been a lot of research on algorithms for computing shortest paths
 using the hub labeling framework (see e.g.\ the following papers by Abraham et al.\ \cite{DBLP:conf/soda/AbrahamFGW10, DBLP:conf/wea/AbrahamDGW11, DBLP:journals/jea/AbrahamDGW13, DBLP:conf/esa/AbrahamDGW12, DBLP:conf/icalp/AbrahamDFGW11, DBLP:conf/gis/AbrahamDFGW12}).
It was noted that these algorithms perform really well in practice (see e.g.\  \cite{DBLP:conf/wea/AbrahamDGW11}). A systematic attempt to explain why this is the case led to the introduction of the notion of \textit{highway dimension}~\cite{DBLP:conf/soda/AbrahamFGW10}. Highway dimension is an interesting concept that managed to explain, at least partially, the success of the above methods: it was proved that graphs with small highway dimension have hub labelings with a small number of hubs; moreover, there is
 evidence that most real-life road networks have low highway dimension \cite{Bast06}.

However, most papers on HL offer only algorithms with absolute guarantees on the cost of the hub labeling they find
(e.g., they show that a graph with a given highway dimension has a hub labeling of a certain size and provide an
algorithm that finds such a hub labeling); they do not relate the cost of the hub labeling to the cost of the optimal hub labeling.
There are very few results on the approximability of the Hub Labeling problem.
Only very recently, Babenko et al.~\cite{DBLP:conf/mfcs/BabenkoGKSW15} and White~\cite{DBLP:conf/esa/White15} proved respectively that \HL1 and \HL\infty are NP-hard. Cohen et al.~\cite{DBLP:journals/siamcomp/CohenHKZ03} gave an $O(\log n)$-approximation algorithm for \HL1 by reducing the problem to a Set Cover instance and using the greedy algorithm for Set Cover to solve the obtained instance
(the latter step is non-trivial since the reduction gives a Set Cover instance of exponential size);
later, Babenko et al.~\cite{DBLP:conf/icalp/BabenkoGGN13} gave a combinatorial $O(\log n)$-approximation algorithm for \HLp, for any $p \in [1, \infty]$.

\paragraph{Our Results.} In this paper, we obtain the following results. We prove an $\Omega(\log{n})$ hardness for \HL1 and \HL\infty on graphs that have multiple shortest paths between some pairs of vertices
(assuming that $\mathrm{P}\neq \mathrm{NP}$). The result shows that the algorithms by Cohen et al.\  and Babenko et al.\  are asymptotically optimal. Since
it is impossible to improve the approximation guarantee of $O(\log n)$ for arbitrary graphs, we focus on special families of graphs.
We consider the family of graphs with \textit{unique shortest paths} --- graphs in which there is only one
shortest path between every pair of vertices. This family of graphs appears in the majority of prior works on hub labeling (see e.g.\  \cite{DBLP:conf/icalp/AbrahamDFGW11, DBLP:conf/mfcs/BabenkoGKSW15, DBLP:conf/esa/AbrahamDGW12}) and is very natural, in our opinion, since in real life all edge lengths are somewhat random, and, therefore, any two paths between two vertices $u$ and $v$ have different lengths. For such graphs, we design an approximation algorithm
with approximation guarantee $O(\log D)$, where $D$ is the shortest-path diameter of the graph (which equals the maximum hop length of a shortest path; see Section~\ref{Prelim-Def} for the definition);
the algorithm works for every fixed $p\in [1,\infty)$ (the constant in the $O$-notation depends on $p$).
In particular, this algorithm gives an $O(\log\log n)$  factor approximation for graphs of diameter $\mathop{\mathrm{polylog}} n$,
while previously known algorithms give only an $O(\log n)$ approximation.
Our algorithm crucially relies on the fact that the input graph has unique shortest paths;
in fact, our lower bounds of $\Omega(\log n)$ on the approximation ratio apply to graphs of constant diameter (with non-unique shortest paths).
We also extensively study HL on trees. Somewhat surprisingly, the problem is not at all trivial on trees. In particular, the standard LP
relaxation for the problem is not integral. We obtain the following results for trees.
\begin{enumerate}
\item Design a polynomial-time approximation scheme (PTAS) for \HLp for every $p\in[1,\infty]$.
\item Design an exact quasi-polynomial time algorithm for \HLp for every $p\in[1,\infty]$, with the running time
 $n^{O(\log n)}$ for $p\in\{1,\infty\}$ and  $n^{O(\log^2 n)}$ for $p\in (1,\infty)$.
\item Analyze a simple combinatorial heuristic for trees, proposed by Peleg in 2000, and prove that it gives a 2-approximation for \HL1
(we also show that this heuristic does not work well for \HLp when $p$ is large).
\end{enumerate}

\paragraph{Organization and overview of results.}
We first present an $O(\log D)$ approximation algorithm for graphs with unique shortest paths (see Sections~\ref{Prelim} and~\ref{Bounded}). The algorithm solves a natural linear programming (LP) relaxation for the problem. Then, it uses the LP solution to find a system of ``pre-hubs'' $\{\widehat H_u\}_{u\in V}$ --- pre-hubs do not necessarily satisfy the covering property (which hubs must satisfy), instead, they satisfy a ``weak'' covering property (see Definition~\ref{def:prehub} for details).
The cost of the pre-hub labeling is at most 2 times the LP value.
Finally, the algorithm converts the pre-hub labeling $\{\widehat H_u\}_{u\in V}$
to a  hub labeling $\{H_u\}_{u\in V}$. To this end, it runs a randomized ``correlated''  rounding scheme (which always finds a feasible hub labeling). The expected cost of the obtained labeling is at most $O(\log D)$ times the cost of the pre-hub labeling. Thus, the algorithm finds an $O(\log D)$-approximate solution. We present the algorithm for \HL1 in Section~\ref{Bounded} and the algorithm for general \HLp in Appendix~\ref{sec:lp-norm-algorithm} (the algorithms for \HL1 and \HLp are almost identical, but the analysis for \HL1 is slightly simpler).

A key property of our rounding procedure that we would like to emphasize is that it might add a vertex $v$ to a hub set $H_u$, even if the corresponding LP indicator variable $x_{uv}$ (which, in an intended integral solution, is 1 if and only if $v \in H_u$) is set to zero in the fractional solution that we consider. This may seem surprising since often rounding schemes satisfy the following natural property: the rounding scheme adds an element to a set or makes an assignment only if the LP indicator variable for the corresponding event is strictly positive. In particular, known rounding schemes for Set Cover, as well as many other thresholding and randomized rounding schemes, satisfy this property (we call such rounding schemes ``natural"). However, we show in Appendix~\ref{lower_bound_rounding} that any rounding scheme that satisfies this property cannot give a better than $\Theta(\log n)$ approximation for \HL1.

Then, in Section~\ref{Trees}, we present our algorithms for trees. We consider a special class of hub labelings --- \textit{hierarchical hub labelings} --- introduced by Abraham et al.~\cite{DBLP:conf/esa/AbrahamDGW12} (we define and briefly discuss hierarchical hub labelings in Section~\ref{hhl-section}). We start with proving that there is an optimal solution for a tree, which is a hierarchical hub labeling. We observe that there is a simple dynamic program (DP) of exponential size for computing the optimal hierarchical hub labeling on a tree: roughly speaking, we have an entry $B[T']$ for every subtree $T'$ of the input tree $T$ in the DP table; entry $B[T']$ equals the cost of the optimal hub labeling  for $T'$; its value can be  computed from the values $B[T'']$ of subtrees $T''$ of $T'$. We then modify the DP and obtain a polynomial-time approximation scheme (PTAS). We do that by restricting the DP table to subtrees $T'$ with a ``small boundary'', proving that the obtained algorithm finds a $(1+\varepsilon)$-approximate solution and bounding its running time. We also describe a quasi-polynomial time algorithm, based on the same DP approach. We first present the algorithm for \HL1  and then a slightly more involved algorithm for any \HLp (see Sections~\ref{ptas_trees}, \ref{quasi_sec}, and Appendices \ref{appendix_ptas_lp}-\ref{appendix_quasi_infinity}). In Section~\ref{trees_2_approx}, we analyze the heuristic for trees proposed by Peleg~\cite{DBLP:journals/jgt/Peleg00}. Using the primal-dual technique, we show that the heuristic finds a 2-approximation for \HL1. We also give an example where the gap between the optimal solution and the solution that the heuristic finds is $3/2 - \varepsilon$. (For $p > 1$, the gap between the optimal solution for \HLp and the solution that the heuristic finds can be at least $\Omega(p/\log p)$.)

Finally, in Section~\ref{Hardness}, we prove an $\Omega(\log n)$-hardness for \HL1 and \HL\infty by constructing a reduction from Set Cover.
Let us very informally describe the intuition behind the proof for \HL1 (the proof for \HL\infty is similar).
Given a Set Cover instance, we construct a graph of constant diameter.
In this graph, we have a special vertex $r$, vertices $S_1,\dots, S_m$, and vertices $x_1,\dots, x_n$ (as well as auxiliary vertices).
Vertices $S_1,\dots, S_m$ correspond to the sets, and vertices $x_1,\dots, x_n$ correspond to
the elements of the universe in the Set Cover instance.
We design the HL instance in such a way that  we  only have  to satisfy the covering property for pairs $(r,x_i)$
(satisfying other pairs of vertices requires very few hubs); furthermore, the cost of the solution is approximately equal to the size of $H_r$
(to guarantee that, we have many copies of $r$ in the graph; their total contribution to the objective is much greater than the contribution of other vertices).
There are multiple shortest paths between $r$ and each $x_i$: if $x_i \in S_j$ in the Set Cover instance, then 
there is a path $r\to S_j \to x_i$. In order to cover the pair $(r, x_i)$ we have to choose one of the paths
between $r$ and $x_i$, then choose a vertex $w$ on it, and add it to $H_r$. The instance is designed in such a way, that if we choose path $r \to S_j\to x_i$, 
then the ``optimal choice'' of $w$ is $S_j$ (in particular, we ensure that all $x_i$ cannot choose $r$ as their hub; we do so by having many copies of vertex $r$
and using auxiliary vertices).
Therefore, to satisfy all pairs $(r, x_i)$, we have to find the smallest possible number of subsets $S_j$ that cover vertices $x_1,\dots, x_n$ and add them
to $H_r$; that is, we need to solve the Set Cover instance.
\\[5pt]
\textit{Note: For completeness, we also include a short note (Appendix~\ref{appendix_directed}) in which we explain how many of our results can be extended to the case of directed graphs.}

\section{Preliminaries}\label{Prelim}
\subsection{Definitions}\label{Prelim-Def}
We will say that a graph $G=(V, E)$ with non-negative edge lengths $l(e)$ has unique shortest paths if there is a unique shortest path between every pair of vertices. (We note that if the lengths of the edges are obtained by measurements, affected by some noise, the graph will satisfy the unique shortest path
property with probability 1.)

Recall that the shortest-path diameter $D$ of a graph $G$ is the maximum hop length of a shortest path in $G$
(that is, it is the minimum number $D$ such that every shortest path contains at most $D$ edges).
Note that $D$ is upper bounded by the aspect ratio $\rho$ of the graph:
$$D\leq \rho\equiv \frac{\max_{u,v\in V} d(u,v)}{\min_{(u,v)\in E} l(u,v)}.$$
Here, $d(u,v)$ is the shortest-path distance in $G$ w.r.t.\ edge lengths $l(e)$.
In particular, if all edges in $G$ have length at least $1$, then $D\leq \mathrm{diam}(G)$, where $\mathrm{diam}(G) = \max_{u,v\in V} d(u,v)$.

We will use the following observation about hub labelings: the covering property for the pair $(u,u)$ requires that $u\in H_u$.

\subsection{Linear programming relaxation}
In this section, we introduce a natural LP formulation of \HL1. Let $I$ be the set of all (unordered) pairs of vertices, including pairs $(u,u)$, which we also denote as $\{u,u\}$, $u \in V$. We use indicator variables $x_{uv}$, for all $(u,v) \in V \times V$, that represent whether $v \in H_u$ or not. Let $S_{uv} (\equiv S_{vu})$ be the set of all vertices that appear in any of the (possibly many) shortest paths between $u$ and $v$ (including the endpoints $u$ and $v$). Note that, although the number of shortest paths between $u$ and $v$ might, in general, be exponential in $n$, the set $S_{uv}$ can always be computed in polynomial time. In case there is a unique shortest path between $u$ and $v$, we use both $S_{uv}$ and $P_{uv}$ to denote the vertices of that unique shortest path. The constraint implied by the covering property is ``$\sum_{w \in S_{uv}} \min\{x_{uw}, x_{vw}\} \geq 1$, for all $\{u,v\} \in I$". The resulting LP relaxation is given below.

\bigskip
\noindent($\mathbf{LP_1}$)
\begin{equation*}
\begin{split}
    \min: &\quad \sum_{u \in V} \sum_{v \in V} x_{uv} \\
    \text{s.t.:} &\quad \sum_{w \in S_{uv}} \min\{x_{uw}, x_{vw}\} \geq 1,  \quad \forall \{u, v\} \in I, \\
          &\quad x_{uv} \geq 0,  \hspace{95pt} \forall (u,v) \in V \times V.
\end{split}
\end{equation*}

\noindent We note that the constraint $\sum_{w \in S_{uv}} \min\{x_{uw}, x_{vw}\} \geq 1$ can be equivalently rewritten as follows: $\sum_{w \in S_{uv}} y_{uvw} \geq 1$, $x_{uw} \geq y_{uvw}$ and $x_{vw}\geq y_{uvw}$. Observe that the total number of variables and constraints remains polynomial, and thus, an optimal solution can always be found efficiently.

One indication that the above LP is indeed an appropriate relaxation of HL is that we can reproduce the result of \cite{DBLP:journals/siamcomp/CohenHKZ03} and get an $O(\log n)$-approximation algorithm for HL by using a very simple rounding scheme. But, we will use the above LP in more refined ways, mainly in conjunction with the notion of pre-hubs, which we introduce later on.

\subsection{Hierarchical hub labeling}\label{hhl-section}
We now define and discuss the notion of hierarchical hub labeling (HHL), introduced by Abraham et al.~\cite{DBLP:conf/esa/AbrahamDGW12}. The presentation in this section follows closely the one in~\cite{DBLP:conf/esa/AbrahamDGW12}.
\begin{definition}
Consider a set system $\{H_u\}_{u\in V}$. Let us say that $v\preceq u$ if $v\in H_u$.
The set system $\{H_u\}_{u\in V}$ is a hierarchical hub labeling if it is a hub labeling, and $\preceq$ is a partial order.
\end{definition}

We will say that $u$ is higher ranked than $v$ if $u\preceq  v$. Every two vertices $u$ and $v$ have a common hub $w\in H_u \cap H_v$, and thus there is a vertex $w$ such that $w\preceq u$ and $w\preceq v$. Therefore, there is the highest ranked vertex in $G$.

We now define a special type of hierarchical hub labelings. Given a total order $\pi: [n] \to V$, a \textit{canonical} labeling is the hub labeling $H$ that is obtained as follows: $v \in H_u$ if and only if $\pi^{-1}(v) \leq \pi^{-1}(w)$ for all $w \in S_{uv}$. It is easy to see that a canonical labeling is a feasible hierarchical hub labeling.
We say that a hierarchical hub labeling $H$ respects a total order $\pi$ if the implied (by $H$) partial order is consistent with $\pi$. Observe that there might be many different total orders that $H$ respects. In~\cite{DBLP:conf/esa/AbrahamDGW12}, it is proved that all total orders that $H$ respects have the same canonical labeling $H'$, and $H'$ is a subset of $H$. Therefore, $H'$ is a minimal hierarchical hub labeling that respects the partial order that $H$ implies.

From now on, all hierarchical hub labelings we consider will be canonical hub labelings. Any canonical hub labeling can be obtained by the following process~\cite{DBLP:conf/esa/AbrahamDGW12}. Start with empty sets $H_u$,
choose a vertex $u_1$ and add it to each hub set $H_u$. Then, choose another vertex $u_2$. Consider all pairs $u$ and $v$ that
currently do not have a common hub, but such that $u_2$ lies on a shortest path between $u$ and $v$. Add $u_2$ to $H_u$ and $H_v$.
Then, choose $u_3$, \dots, $u_n$, and perform the same step. We get a hierarchical hub labeling.
(The hub labeling, of course, depends on the order in which we choose vertices of $G$.)

This procedure is particularly simple if the input graph is a tree (we will use this in Section \ref{Trees}). In a tree, we choose a vertex $u_1$ and add it to each hub set $H_u$. We remove $u_1$ from the tree and recursively process each connected component of $G-u_1$.
No matter how we choose vertices $u_1,\dots, u_n$, we get a canonical hierarchical hub labeling; given a hierarchical hub labeling $H$, in order to get a canonical hub labeling $H'$,
we need to choose the vertex $u_i$ of highest rank in $T'$ (w.r.t. to the order $\preceq$ defined by $H$) when our recursive procedure processes subinstance $T'$.
A canonical hub labeling gives a recursive decomposition of the tree to subproblems of gradually smaller size.
\section{Pre-hub labeling}
We introduce the notion of a pre-hub labeling that we will use in designing algorithms for HL.
\begin{definition}[Pre-hub labeling]\label{def:prehub}
Consider a graph $G = (V,E)$ and a length function $l: E \to \mathbb{R}^+$; assume that all shortest paths are unique. A family of sets $\{\widehat{H}_u\}_{u \in V}$, with $\widehat{H}_u \subseteq V$, is called a pre-hub labeling, if for every pair $\{u,v\}$, there exist $u' \in \widehat{H}_u \cap P_{uv}$ and $v' \in \widehat{H}_v \cap P_{uv}$ such that $u' \in P_{v'v}$; that is, vertices $u$, $v$, $u'$, and $v'$ appear in the following order along $P_{uv}$:
$u, v', u', v$ (possibly, some of the adjacent, with respect to this order, vertices coincide).
\end{definition}
Observe that any feasible HL is a valid pre-hub labeling. We now show how to find a pre-hub labeling given a feasible LP solution.
\begin{lemma}\label{lem:prehubs}
Consider a graph $G = (V,E)$ and a length function $l: E \to \mathbb{R}^+$; assume that all shortest paths are unique.
Let $\{x_{uv}\}$ be a feasible solution to $\mathbf{LP_1}$. Then, there exists a pre-hub labeling $\{\widehat{H}_u\}_{u\in V}$ such that $|\widehat{H}_u| \leq 2\sum_{v\in V} x_{uv}$.
In particular, if $\{x_{uv}\}$ is an optimal LP solution and $OPT$ is the $\ell_1$-cost of the optimal hub labeling (for \HL1), then
$\sum_{u\in V}|\widehat{H}_u| \leq 2 \, OPT$.
Furthermore, the pre-hub labeling $\{\widehat{H}_u\}_{u\in V}$ can be constructed efficiently given the LP solution $\{x_{uv}\}$.
\end{lemma}
\begin{proof}
Let us fix a vertex $u$.
We build the breadth-first search tree $T_u$ (w.r.t.\ edge lengths; i.e. the shortest-path tree)  from $u$; tree $T_u$ is rooted at $u$ and contains those edges $e \in E$ that appear on a shortest path between $u$ and some vertex $v \in V$. Observe that $T_u$ is indeed a tree and is uniquely defined, since we have assumed that shortest paths in $G$ are unique. For every vertex $v$, let $T'_{uv}$ be the subtree of $T_u$ rooted at vertex $v$.
Given a feasible LP solution $\{x_{uv}\}$, we define the weight of $T'_{uv}$ to be $\mathcal{W}(T'_{uv}) = \sum_{w \in T'_{uv}} x_{uw}$.

We now use the following procedure to construct set $\widehat{H}_{u}$.
We process the tree $T_u$ bottom up (i.e.\ we process a vertex $v$ after we have processed all other vertices in the subtree rooted at $v$), and whenever we detect a subtree $T'_{uv}$ of $T_u$ such that $\mathcal{W}(T'_{uv}) \geq 1/2$, we add vertex 
$v$ to the set $\widehat{H}_u$. We then set $x_{uw} = 0$ for all $w \in T'_{uv}$, and continue (with the updated $x_{uw}$ values) until we reach the root $u$ of $T_u$.
Observe that every time we add one vertex to $\widehat{H}_u$, we decrease the value of $\sum_{v \in V} x_{uv}$ by at least $1/2$.
Therefore, $|\widehat{H}_u| \leq 2 \cdot \sum_{v \in V} x_{uv}$.
We will now show that sets $\{\widehat{H}_u\}$ form a pre-hub labeling. To this end, we prove the following two claims.
\begin{claim}\label{claim:prehub-1}
Consider a vertex $u$ and two vertices $v_1, v_2$ such that $v_1\in P_{uv_2}$. If $\widehat{H}_u \cap P_{v_1 v_2} = \varnothing$, then
$\sum_{w\in P_{v_1v_2}} x_{uw} < 1/2$.
\end{claim}
\begin{proof}
Consider the execution of the algorithm that defined $\widehat{H}_u$. Consider the moment $M$ when we processed vertex $v_1$.
Since we did not add $v_1$ to $\widehat{H}_u$, we had $\mathcal{W}(T'_{uv_1}) < 1/2$.
In particular, since $P_{v_1v_2}$ lies in $T'_{uv_1}$, we have $\sum_{w\in P_{v_1v_2}} x_{uw}'<1/2$,
where $x_{uw}'$ is the value of $x_{uw}$ at the moment $M$. Since none of the vertices on the path $P_{v_1v_2}$ were added to $\widehat{H}_u$, none of the variables $x_{uw}$ for $w\in P_{v_1v_2}$ had been set to $0$.
Therefore, $x_{uw}' = x_{uw}$ (where $x_{uw}$ is the initial value of the variable) for $w\in P_{v_1v_2}$.
We conclude that $\sum_{w\in P_{v_1v_2}} x_{uw}<1/2$, as required.
\end{proof}

\begin{claim}\label{claim:prehub-2}
For any shortest path $P_{uv}$, let $u' \in \widehat{H}_u \cap P_{uv}$ be the vertex closest to $v$ among all vertices in $\widehat{H}_u \cap P_{uv}$ and $v' \in \widehat{H}_v \cap P_{uv}$ be the vertex closest to $u$ among all vertices in $\widehat{H}_v \cap P_{uv}$. Then $u' \in P_{v'v}$.
(Note that $\widehat{H}_u \cap P_{uv}\neq \varnothing$, since always $x_{uu} = 1$ and hence $u\in \widehat{H}_u \cap P_{uv}$; similarly,
$\widehat{H}_v \cap P_{uv}\neq \varnothing$.)
\end{claim}
\begin{proof}
Let us assume that this is not the case; that is, $u' \notin P_{v'v}$. Then $v'\neq u$ and $u'\neq v$ (otherwise, we would trivially have $u' \in P_{v'v}$).
Let $u''$ be the first vertex after $u'$ on the path $P_{u'v}$, and $v''$ be the first vertex after $v'$ on the path $P_{v'u}$.
Since $u' \notin P_{v'v}$, every vertex of $P_{uv}$ lies either on $P_{v''u}$ or $P_{u''v}$, or both (i.e. $P_{v''u} \cup P_{u''v} = P_{uv}$).

By our choice of $u'$, there are no pre-hubs for $u$ on $P_{u''v}$. By Claim~\ref{claim:prehub-1},
 $\sum_{w\in P_{u''v}} x_{uw} <1/2$. Similarly, $\sum_{w\in P_{v''u}} x_{vw} <1/2$.
Thus,
$
    1 > \sum_{w \in P_{uv''}} x_{vw} + \sum_{w \in P_{u''v}} x_{uw} \geq \sum_{w \in P_{uv}} \min\{x_{uw}, x_{vw}\}
$. 
We get a contradiction since $\{x_{uv}\}$ is a feasible LP solution.
\end{proof}
Claim~\ref{claim:prehub-2} shows that $\{\widehat{H}_u\}$ is a valid pre-hub labeling.
\end{proof}

\section{Hub labeling on graphs with unique shortest paths}\label{Bounded}
In this section, we present an $O(\log D)$-approximation algorithm for \HLp on graphs with unique shortest paths,
where $D$ is the shortest-path diameter of the graph. The algorithm works for every fixed $p \in [1,\infty)$ (the hidden constant factor in  the approximation factor  $O(\log D)$  depends on $p$).
To simplify the exposition, we present the algorithm for \HL1 in this section,
and the algorithm for \HLp, for arbitrary fixed $p \geq 1$, in Appendix~\ref{sec:lp-norm-algorithm}.

Consider Algorithm~\ref{Pre-Hubs_Algorithm} in the figure. The algorithm solves the LP relaxation and computes a pre-hub labeling $\{\widehat{H}_u\}_{u \in V}$
as described in Lemma~\ref{lem:prehubs}. Then it chooses a random permutation $\pi$ of $V$ and goes over all vertices one-by-one in the order specified by $\pi$:
$\pi_1$, $\pi_2$,\dots, $\pi_n$. It adds $\pi_i$ to $H_u$ if there is a pre-hub $u'\in \widehat{H}_u$ such that the following conditions hold:
 $\pi_i$ lies on the path $P_{uu'}$, there are no pre-hubs for $u$ between $\pi_i$ and $u'$ (other than $u'$), and currently there are no hubs for $u$  between
 $\pi_i$ and $u'$.

\begin{algorithm}[h]
    Solve $\mathbf{LP_1}$ and get an optimal solution $\{x_{uv}\}_{(u,v) \in V \times V}$.\\
    Obtain a set of pre-hubs $\{\widehat{H}_u\}_{u \in V}$ from $x$ as described in Lemma~\ref{lem:prehubs}. \\
    Generate a random permutation $\pi : [n] \to V$ of the vertices. \\
    Set $H_u = \varnothing$, for every $u \in V$. \\
    \For{$i= 1$ \KwTo $n$} {
        \For{every $u \in V$} {
            \For{every $u' \in \widehat{H}_u$, such that $\pi_i \in P_{uu'}$ and $P_{\pi_i u'} \cap \widehat{H}_u = \{u'\}$} {
                \lIf{$P_{\pi_i u'} \cap H_u = \varnothing$}{$H_u := H_u \cup \{\pi_i\}$}
            }
        }
    }
    Return $\{H_u\}_{u \in V}$.
    \medskip
\caption{Algorithm for \HL1 on graphs with unique shortest paths}
\label{Pre-Hubs_Algorithm}
\end{algorithm}

\begin{theorem}
Algorithm~\ref{Pre-Hubs_Algorithm} always returns a feasible hub labeling $H$. The cost of the hub labeling  is
$\E[\sum_u |H_u|] = O(\log D) \cdot OPT_{LP_1}$ in expectation, where $OPT_{LP_1}$ is the optimal value of $\mathbf{LP_1}$.
\end{theorem}
\begin{remark}\label{rem:derandomize}
Note that Algorithm~\ref{Pre-Hubs_Algorithm} can be easily derandomized using the method of conditional expectations: instead of 
choosing a random permutation $\pi$, we first choose $\pi_1\in V$, then $\pi_2\in V\setminus\{\pi_1\}$ and so on;
each time we choose $\pi_i \in V\setminus\{\pi_1,\dots,\pi_{i-1}\}$ so as to minimize the conditional expectation 
$\E[\sum_u |H_u| \,|\, \pi_1,\dots, \pi_i]$.
\end{remark}
\begin{proof}
We first show that the algorithm always finds a feasible hub labeling. Consider a pair of vertices $u$ and $v$. We need to show that they have a common hub
on $P_{uv}$. The statement is true if $u=v$ since $u\in \widehat{H}_u$ and thus $u\in H_u$.
So, we assume that $u \neq v$. Consider the path $P_{uv}$. Because of the pre-hub property, there exist $u' \in \widehat{H}_u$ and $v' \in \widehat{H}_v$ such that $u' \in P_{v'v}$.
In fact, there may be several possible ways to choose such $u'$ and $v'$. Let us choose $u'$ and $v'$ so that $\widehat{H}_u \cap (P_{u'v'} \setminus \{u',v'\}) = \widehat{H}_v \cap (P_{u'v'} \setminus \{u',v'\}) = \varnothing$ (for instance, choose the closest pair of  $u'$ and $v'$ among all possible pairs).
Consider the first iteration $i$ of the algorithm such that $\pi_i \in P_{u'v'}$. We claim that the algorithm adds $\pi_i$ to both $H_u$ and $H_v$.
Indeed, let us verify that the algorithm adds $\pi_i$ to $H_u$. We have: (i)  $\pi_i$ lies on $P_{v'u'}\subset P_{uu'}$, (ii) there are no pre-hubs for $u$ on $P_{v'u'}\supset P_{\pi_iu'}$ other than $u'$,
 (iii) $\pi_i$ is the first vertex we process on the path $P_{u'v'}$, thus currently there are no hubs on $P_{u'v'}$.
Therefore, the algorithm adds $\pi_i$ to $H_u$. Similarly, the algorithm adds $\pi_i$ to $H_v$.

Now we upper bound the expected cost of the solution. We will charge every hub that we add to $H_u$ to a pre-hub in $\widehat{H}_u$;
namely, when we add $\pi_i$ to $H_u$ (see line 8 of the algorithm), we charge it to pre-hub $u'$.
For every vertex $u$, we have $|\widehat{H}_u| \leq 2\sum_w x_{uw}$. We are going to show that every
$u' \in \widehat{H}_u$ is charged at most $O(\log D)$ times in expectation. Therefore, the expected number of hubs in $H_u$ is at most
$O(2\sum_w x_{uw} \times \log D)$.

Consider a vertex $u$ and a pre-hub $u'\in \widehat{H}_u$ ($u'\neq u$). Let $u'' \in \widehat{H}_u$ be the closest pre-hub to $u'$ on the path $P_{u'u}$.
Observe that all hubs charged to $u'$ lie on the path $P_{u''u'} \setminus \{u''\}$.
Let $k = |P_{u''u'} \setminus \{u''\}|$. Note that $k \leq D$. Consider the order $\sigma : [k] \to P_{u''u'} \setminus \{u''\}$ in which the vertices of $P_{u''u'} \setminus \{u''\}$ were processed by the algorithm ($\sigma$ is a random permutation). Note that $\sigma_i$ charges $u'$ if and only if $\sigma_i$ is closer to $u'$ than $\sigma_{1},\dots,\sigma_{i-1}$.
The probability of this event is $1/i$. We get that the number of hubs charged to $u'$ is $\sum_{i=1}^k \frac{1}{i} = \log k + O(1)$, in expectation.  Hence, $\E \left[\sum_{u\in V} |H_u| \right] \leq 2 \left(\log D + O(1) \right) \cdot OPT_{LP_1}$.
\end{proof}

\section{Hub labeling on trees}\label{Trees}

In this section, we study the Hub Labeling problem on trees. Observe that if the graph is a tree, the length function $l$ does not play any role in the task of choosing the optimal hubs (it only affects the actual distances between the vertices), and so we assume that we are simply given an unweighted tree $T = (V, E)$, $|V| = n$. We start with proving a structural result about optimal solutions in trees; we show that there always exists a \textit{hierarchical hub labeling} that is also an optimal hub labeling.
We then analyze a simple and fast heuristic for HL on trees proposed by Peleg \cite{DBLP:journals/jgt/Peleg00}, and prove that it gives a 2-approximation for \HL1. We do not know if our analysis is tight, but we prove that there are instances where the heuristic finds a suboptimal solution of cost at least $(\frac{3}{2}-\varepsilon) OPT$ (for every $\varepsilon > 0$).
Finally, we present a polynomial-time approximation scheme (PTAS) and a quasi-polynomial time exact algorithm for \HL1 on trees. We then modify our approach in order to obtain a PTAS and a quasi-polynomial-time exact algorithm for \HLp on trees for any $p \in [1, \infty]$, thus providing a thorough algorithmic understanding of HL, under any cost function, on trees.


\subsection{Optimal solutions for trees are hierarchical}
We show that any feasible hub labeling solution $H$ can be converted to a hierarchical hub labeling $H'$ of at most the same $\ell_p$-cost
(for every $p\in[1,\infty]$). Therefore, there always exists an optimal solution that is hierarchical.

\begin{theorem}\label{HHL_optimal_for_trees}
For every tree $T = (V,E)$, there always exists an optimal \HLp solution that is hierarchical, for every $p \in [1, \infty]$.
\end{theorem}
\begin{proof}
To prove this, we consider a feasible solution $H$ and convert it to a hierarchical solution $H'$ such that $|H'_u|\leq |H_u|$ for every $u$.
In particular, the $\ell_p$-cost of $H'$ is at most the $\ell_p$-cost of $H$ for every $p$.

The construction is recursive (the underlying inductive hypothesis for smaller subinstances being that a feasible HL $H$ can be converted to a hierarchical solution $H'$ such that $|H'_u|\leq |H_u|$ for every $u$.) First, for each $u \in V$, define an induced subtree $T_u \subseteq T$ as follows: $T_u$ is the union of paths $P_{uv}$ over all $v\in H_u$. In other words, a vertex $w$ belongs to $H_u$ if there is a hub $v\in H_u$ such that $w\in P_{uv}$. Note that $T_u$ is a connected subtree of $T$.

The crucial property that we need is that $T_u \cap T_v \neq \varnothing$, for every $u,v \in V$. To see this, consider any pair $\{u,v\}$, $u \neq v$.  We know that $H_u \cap H_v \cap P_{uv} \neq \varnothing$. Let $w \in H_u \cap H_v \cap P_{uv}$. By construction, $w \in T_u$ and $w \in T_v$, and so $T_u \cap T_v \neq \varnothing$. We now use the fact that a family of subtrees of a tree satisfies the \textit{Helly property} (which first appeared as a consequence of the work of Gilmore \cite{Gilmore}, and more explicitly a few years later in \cite{gyarfas1970helly}): suppose that we are given a family of subtrees of $T$ such that every two subtrees in the family intersect, then all subtrees in the family intersect (i.e. they share a common vertex).

Let $r \in \bigcap_{u \in V} T_u$. We remove $r$ from $T$. Consider the connected components of $T-r$: $Q_1,\dots,Q_c$. Denote the connected component that contains vertex $u$ by $Q^u$. Let $\widetilde H_u = H_u\cap Q^u$. Note that $|\widetilde H_u|\leq|H_u| -1$, since $r \in T_u$, which, by the definition of $T_u$, implies that there exists some $w \notin Q^u$ with $w \in H_u$.
Consider $u,v\in Q_i$. They have a common hub $w\in H_u \cap H_v\cap P_{uv}$. Since $P_{uv} \subset Q^u = Q^v = Q_i$, we have $w\in \widetilde H_u \cap \widetilde H_v \cap P_{uv}$. Therefore, $\{\widetilde H_u:u\in Q_i\}$ is a feasible hub labeling for $Q_i$.
Now, we recursively find hierarchical hub labelings for $Q_1,\dots, Q_c$. Denote the hierarchical hub labeling for $u$ in $Q^u$ by $H''_u$. The inductive hypothesis ensures that $|H''_u| \leq |\widetilde H_u|\leq |H_u| -1$.

Finally, define $H_u' = H''_u \cup \{r\}$, for $u\neq r$, and $H_r' = \{r\}$. We show that $H_u'$ is a hub labeling. Consider $u,v\in V$.
If $u,v\in Q_i$ for some $i$, then
$H'_u \cap H'_v\cap P_{uv} \supset H''_u \cap H''_v\cap P_{uv}\neq \varnothing$ since $H''$ is a hub labeling for $Q_i$.
If $u\in Q_i$ and $v \in Q_j$ ($i\neq j$), then $r\in  H'_u \cap H'_v\cap P_{uv}$. Also, if either $u=r$ or $v=r$, then again
$r\in H'_u \cap H'_v\cap P_{uv}$. We conclude that $H'$ is a hub labeling. Furthermore,
$H'$ is a \textit{hierarchical} hub labeling: $r\preceq u$ for every $u$ and $\preceq$ is a partial order on every set $Q_i$;
elements from different sets $Q_i$ and $Q_j$ are not comparable w.r.t.\ $\preceq$.

We have $|H_u'| = |H''_u| + 1 \leq |H_u|$ for $u\neq r$ and $|H'_r| = 1\leq |H_r|$, as required.
\end{proof}

This theorem allows us to restrict our attention only to hierarchical hub labelings, which have a much simpler structure than arbitrary hub labelings, when we design algorithms for HL on trees.

\subsection{An analysis of Peleg's heuristic for \texorpdfstring{\HL1}{HL1} on trees} \label{trees_2_approx}

In this section, we  analyze a purely combinatorial algorithm for HL proposed by Peleg in \cite{DBLP:journals/jgt/Peleg00} and show that it returns a hierarchical 2-approximate hub labeling on trees (see Algorithm~\ref{Tree-Algorithm}). (Peleg's result~\cite{DBLP:journals/jgt/Peleg00}, expressed in the hub labeling context, shows that the algorithm returns a hub labeling $H$ with $\max_{u \in V} |H_u| = O(\log n)$ for a tree on $n$ vertices.)

\begin{definition}\label{def:bal-sep-vertex}
Consider a tree $T$ on $n$ vertices. We say that a vertex $u$ is a balanced separator vertex if every connected component
of $T - u$  has at most $n/2$ vertices. The weighted balanced separator vertex for a vertex-weighted tree is defined analogously.
\end{definition}
It is well-known that every tree $T$ has a balanced separator vertex (in fact, a tree may have either exactly one or exactly two balanced separator vertices) and such a separator vertex can be found efficiently (i.e. in linear time) given $T$. The algorithm by Peleg, named here Tree Algorithm, is described in the figure below (Algorithm \ref{Tree-Algorithm}).

\begin{algorithm}[h]
\SetKwInput{Input}{Input}
\SetKwInput{Output}{Output}
        \Input{a tree $T'$}
        \Output{a hub labeling $H$ for $T'$}
\smallskip
        Find a balanced separator vertex $r$ in $T'$. \\
        Remove  $r$ and recursively find a HL in each subtree $T_i$ of $T'-r$. Let $H'$ be the labeling obtained by the recursive procedure.
        (If $T'$ consists of a single vertex and, therefore, $T'-r$ is empty, the algorithm does not make any recursive calls.)
        \\
        Return $H_u := H_u' \cup \{r\}$, for every vertex $u$ in $ T'- \{r\}$, and $H_r = \{r\}$.
        \medskip

\caption{Tree Algorithm}
\label{Tree-Algorithm}
\end{algorithm}

It is easy to see that the algorithm always returns a feasible hierarchical hub labeling, in total time $O(n \log n)$. To bound its cost, we use the primal--dual approach. We consider the dual of $\mathbf{LP_1}$. Then, we define a dual feasible solution whose cost is at least half of the cost of the solution that the algorithm returns. We formally prove the following theorem.
\begin{theorem}\label{peleg_theorem}
	The Tree Algorithm is a 2-approximation algorithm for \HL1 on trees.
\end{theorem}
\begin{proof}
The primal and dual linear programs for HL on trees are given in Figure~\ref{primal-dual}. We note that the dual variables $\{a_{uv}\}_{u,v}$ correspond to unordered pairs $\{u,v\} \in I$, while the variables $\{\beta_{uvw}\}_{u,v,w}$ correspond to ordered pairs $(u,v) \in V \times V$, i.e. $\beta_{uvw}$ and $\beta_{vuw}$ are different variables.

\begin{figure}
\begin{minipage}{0.4\textwidth}
\small
\vspace{-6pt}
(\textbf{PRIMAL-LP})
\vspace{12pt}
\begin{align*}
    \min: &\quad \sum_{u \in V} \sum_{v \in V} x_{uv} \\
    \text{s.t.:}  &\quad \sum_{w \in P_{uv}} y_{uvw} \geq 1, &&\forall\, \{u,v\} \in I  \\
          &\quad x_{uw} \geq y_{uvw}, &&\forall\, \{u,v\} \in I, \; \forall\, w \in P_{uv} \\
          &\quad x_{vw} \geq y_{uvw}, &&\forall\, \{u,v\} \in I, \; \forall\, w \in P_{uv} \\
          &\quad x_{uv} \geq 0, &&\forall\, \{u,v\} \in V \times V\\
          &\quad y_{uvw} \geq 0, &&\forall\, \{u,v\} \in I, \;\forall\, w \in P_{uv}
\end{align*}
\vspace{8pt}
\end{minipage}%
\hspace{4mm}\vline \hspace{4mm}
\begin{minipage}{0.4\textwidth}
\small
(\textbf{DUAL-LP})\\
variables: $\alpha_{uv}$ and $\beta_{uvw}$ for $w\in P_{uv}$
\begin{align*}
    \max: &\quad \sum_{\{u,v\} \in I} \alpha_{uv} \\
    \text{s.t.:} &\quad \alpha_{uv} \leq \beta_{uvw} + \beta_{vuw}, && \forall\, \{u,v\} \in I \;, u \neq v \\
    & &&\forall\, w \in P_{uv} \\
    &\quad \alpha_{uu} \leq \beta_{uuu}, &&\forall\, u \in V \\
    &\quad \sum_{v: w \in P_{uv}} \beta_{uvw} \leq 1, &&\forall\, (u,w) \in V \times V \\
    &\quad  \alpha_{uv} \geq 0, && \forall\, \{u,v\} \in I \\
    &\quad  \beta_{uvw} \geq 0, && \forall\, \{u,v\} \in I, \forall\, w \in P_{uv} \\
    &\quad  \beta_{vuw} \geq 0,&& \forall\, \{u,v\} \in I, \forall\, w \in P_{uv}
\end{align*}
\end{minipage}
\begingroup
\captionof{table}{Primal and Dual LPs for HL on trees}
\label{primal-dual}
\endgroup
\end{figure}
\normalsize

As already mentioned, it is straightforward that the algorithm finds a feasible hierarchical hub labeling. We now bound the cost of the solution by
 constructing a fractional solution for the DUAL-LP. To this end, we track the execution of the algorithm and gradually define the fractional solution.
Consider one iteration of the algorithm in which the algorithm processes a tree $T'$ ($T'$ is a subtree of $T$).
Let $r$ be the balanced separator vertex that the algorithm finds in line 1.
At this iteration, we assign dual variables $a_{uv}$ and $\beta_{uvw}$ for those pairs $u$ and $v$ in $T'$ for which $P_{uv}$ contains vertex $r$.
Let $n'$ be the size of $T'$, $A = 2/n'$ and $B=1/n'$.
Denote the connected components of $T'-r$ by $T_1, ..., T_t$; each $T_i$ is a subtree of $T'$.

Observe that we assign a value to each $a_{uv}$ and $\beta_{uvw}$ exactly once.
Indeed, since we split $u$ and $v$ at some iteration, we will assign a value to $a_{uv}$ and $\beta_{uvw}$ at least once.
Consider the first iteration in which we assign a value to $a_{uv}$ and $\beta_{uvw}$. At this iteration, vertices $u$ and $v$ lie in different subtrees $T_i$ and $T_j$ of $T'$
(or $r\in\{u,v\}$). Therefore, vertices $u$ and $v$ do not lie in the same subtree $T''$ in the consecutive iterations;
consequently, we will not assign new values to $a_{uv}$ and $\beta_{uvw}$ later.

For $u\in T_i$ and $v\in T_j$ (with $i\neq j$), we define $\alpha_{uv}$ and $\beta_{uvw}$ as follows
\begin{itemizeTwoCol}
	\item $\alpha_{uv} = A$,
    \item For $w \in P_{ur} \setminus \{r\}$: $\beta_{uvw} = 0$ and $\beta_{vuw} = A$.
    \item For $w \in P_{rv} \setminus \{r\}$: $\beta_{uvw} = A$ and $\beta_{vuw} = 0$.
    \item For $w = r$: $\beta_{uvr} =\beta_{vur} = B$.
\end{itemizeTwoCol}
For $u\in T_i$ and $v=r$, we define $\alpha_{uv}$ and $\beta_{uvw}$ as follows
\begin{itemizeTwoCol}
    \item $\alpha_{ur} = A$.
    \item For $w \in P_{ur} \setminus \{r\}$: $\beta_{urw} = 0$ and $\beta_{ruw} = A$.
    \item For $w = r$: $\beta_{urr} = \beta_{rur} = B$.
\end{itemizeTwoCol}
Finally, we set $\alpha_{rr} = \beta_{rrr} = B$.

We now show that the obtained solution $\{\alpha, \beta\}$ is a feasible solution for DUAL-LP.
Consider the first constraint: $\alpha_{uv}\leq \beta_{uvw}+\beta_{vuw}$. If $u\neq r$ or $v\neq r$, $A=\alpha_{uv} = \beta_{uvw} + \beta_{vuw}  = 2B$.
The second constraint is satisfied since $\alpha_{rr} =\beta_{rrr}$.

Let us verify now that the third constraint, $\sum_{v:w\in P_{uv}} \beta_{uvw} \leq 1$, is satisfied.
Consider a \textit{non-zero} variable $\beta_{uvw}$ appearing in the sum. Consider the iteration of the algorithm in which we assign $\beta_{uvw}$
a value. Let $r$ be the balanced separator vertex during this iteration.
Then, $r\in P_{uv}$ (otherwise, we would not assign any value to $\beta_{uvw}$) and
$w\in P_{rv}$. Therefore,  $r\in P_{uw}$; that is, the algorithm assigns the value to $\beta_{uvw}$ in the iteration when it splits $u$ and $w$
(the only iteration when $r\in P_{uw}$). In particular, we assign a value to all non-zero variables $\beta_{uvw}$ appearing in the constraint
in the same iteration of the algorithm. Let us consider this iteration.

If $u\in T_i\cup \{r\}$ and $w\in T_j$, then every $v$ satisfying $w\in P_{uv}$ lies in $T_j$. For every such $v$, we have $\beta_{uvw} = A$. Therefore,
$$\sum_{v:w\in P_{uv}} \beta_{uvw} \leq |T_j| \cdot A \leq \frac{n'}{2} \cdot \frac{2}{n'} = 1,$$
as required.
If $u\in T_i\cup\{r\}$ and $w=r$, then we have
$$\sum_{v: w\in P_{uv}} \beta_{uvw} = \sum_{v: r\in P_{uv}} \beta_{uvr} = \sum_{v: r \in P_{uv}} B \leq n' B   =  1,$$
as required.
We have showed that $\{\alpha,\beta\}$ is a feasible solution. Now we prove that its value is at least half of the value of the hub labeling found by the algorithm.
Since the value of any feasible solution of DUAL-LP is at most the cost of the optimal hub labeling, this will prove that the algorithm gives a 2-approximation.

We consider one iteration of the algorithm. In this iteration, we add $r$ to the hub set $H_u$ of every vertex $u\in T'$. Thus, we increase the cost of the hub labeling by $n'$. We are going to show that the dual variables that we set during this iteration contribute at least $n'/2$ to the value of DUAL-LP.

Let $k_i = |T_i| \leq  n' / 2$, for all $i \in \{1,\dots, t\}$. We have $\sum_i k_i  = n'-1$.
The contribution $C$ of the variables $\alpha_{uv}$ that we set during this iteration to the objective function equals
$$
C = \sum_{i<j} \sum_{u\in T_i, v\in T_j} \alpha_{uv} + \sum_{i} \sum_{u\in T_i} \alpha_{ur} + \alpha_{rr}
= A\sum_{i<j} k_i k_j + A (n'-1) + B
 = \frac{2}{n'} \sum_{i<j} k_i k_j + \frac{2n'-1}{n'}.
$$
Now, since $\sum_{j:j\neq i} k_j = (n'-1 - k_i)\geq (n'-2)/2$, we have
$$\frac{2}{n'} \sum_{i<j} k_i k_j =\frac{1}{n'} \sum_{i\neq j} k_i k_j = \frac{1}{n'}\sum_{i}k_i \left(\sum_{j:j\neq i} k_j\right) \geq
\frac{n'-2}{2n'} \sum_{i} k_i = \frac{(n'-1)(n'-2)}{2n'}.$$
Thus,
$$C \geq \frac{(4n'-2) + (n'^2 -3 n' + 2)}{2n'} = \frac{n'+1}{2}.$$
We proved that $C \geq n'/2$. This finishes the proof.
\end{proof}

Given the simplicity of the Tree Algorithm, it is interesting to understand whether the 2 approximation factor is tight or not. We do not have a matching
lower bound, but we show an asymptotic lower bound of $3/2$. The instances that give this $3/2$ lower bound are the complete binary trees.
We present the proof of the following lemma in Appendix~\ref{proof-of-tree-alg-lower-bound}.
\begin{lemma}\label{tree-alg-lower-bound}
The approximation factor of the Tree Algorithm is at least $3/2 - \varepsilon$, for every fixed $\varepsilon > 0$.
\end{lemma}

\begin{remark}
The Tree Algorithm does not find a good approximation for the $\ell_p$-cost, when $p$ is large. Let $k>1$ be an integer.
Consider a tree $T$ defined as follows:
it consists of a path $a_1,\dots, a_k$ and leaf vertices connected to the vertices of the path; vertex $a_i$ is connected to
$2^{k-i}-1$ leaves. The tree consists of $n=2^k-1$ vertices. It is easy to see that the Tree Algorithm will first choose vertex $a_1$,
then it will process the subtree of $T$ that contains $a_k$ and will choose $a_2$, then $a_3$ and so on. Consequently,
the hub set $H_{a_k}$ equals $\{a_1,\dots, a_k\}$ in the obtained hub labeling. The $\ell_p$-cost of this hub labeling is greater than $k$.
However, there is a hub labeling $\widetilde H$ for the path $a_1,\dots, a_k$ with $|\widetilde H_{a_i}| \leq O(\log k)$, for all $i \in [k]$.
This hub labeling can be extended to a hub labeling of $T$, by letting $\widetilde H_{l} = \widetilde H_{a_i} \cup \{l\}$ for each leaf $l$ adjacent to a vertex $a_i$.
Then we still have $|\widetilde H_u| \leq O(\log k)$, for any vertex $u \in T$.
 The $\ell_p$-cost of this solution
is $O(n^{1/p} \log k)$. Thus, for $k=p$, the gap between the solution $H$ and the optimal solution is at least $\Omega(p/\log p)$.
For $p=\infty$, the gap is at least $\Omega(\log n/\log\log n)$, asymptotically.
\end{remark}


\subsection{A PTAS for HL on trees}\label{ptas_trees}

We now present a polynomial-time approximation scheme (PTAS) for HL on trees. For simplicitly, we only present the algorithm for \HL1, based on dynamic programming (DP). A modified DP approach along similar lines can also be used in the case of \HLp, for any $p \in [1, \infty]$. More details and formal proofs can be found in Appendix~\ref{appendix_ptas_lp} (for \HLp, for every $p \in [1, +\infty)$), and Appendix~\ref{appendix_ptas_infty} (for \HL\infty).

Theorem~\ref{HHL_optimal_for_trees} shows that we can restrict our attention only to hierarchical hub labelings.
That is, we can find an optimal solution by choosing an appropriate vertex $r$, adding $r$ to every $H_u$, and then recursively solving HL on each connected component of $T-r$ (see Section~\ref{hhl-section}). Of course, we do not know what vertex $r$ we should choose, so to implement this approach, we use dynamic programming (DP).
Let us first consider a very basic dynamic programming solution. We store a table $B$ with an entry $B[T']$ for every subtree $T'$ of $T$.
Entry $B[T']$ equals the cost of the optimal hub labeling for tree $T'$. Now if $r$ is the common  hub of all vertices in $T'$, we have
$$B[T'] = |T'|+ \sum_{T'' \text{ is a connected component of } T'-r} B[T'']$$
(the term $|T'|$ captures the cost of adding $r$ to each $H_u$). We obtain the following recurrence formula for the DP:
\begin{equation}\label{eq:DP-main}
B[T'] = |T'|+ \min_{r \in T'} \sum_{T'' \text{ is a connected component of } T'-r} B[T''].
\end{equation}
The problem with this approach, however, is that a tree may have exponentially many subtrees, which means that the size of the dynamic program and the running time
may be exponential.

To work around this, we will instead store $B[T']$ only for some subtrees $T'$, specifically for subtrees with a ``small boundary''. For each subtree $T'$ of $T$, we define its boundary $\partial(T')$ as $\partial(T'): = \{v \notin T': \exists u \in T' \textrm{ with } (u,v) \in E\}$.
Consider now a subtree $T'$ of $T$ and its boundary $S= \partial(T')$.
Observe that if $|S| \geq 2$, then the set $S$ uniquely identifies the subtree $T'$: $T'$ is the unique connected component of $T-S$ that has all vertices from $S$ on its boundary
(every other connected component of  $T-S$ has only one vertex from $S$ on its boundary). If $|S|=1$, that is, $S = \{u\}$ for some $u \in V$,
then it is easy to see that $u$ can serve as a boundary point for $\deg(u)$ different subtrees.

Fix $\varepsilon < 1$. Let $k = 4 \cdot \lceil 1 / \varepsilon \rceil$.
In our dynamic program, we only consider subtrees $T'$ with $|\partial(T')| \leq k$.
Then, the total number of entries is upper bounded by
$\sum_{i = 2}^k \binom{n}{i} + \sum_{u\in V} \deg(u) = O(n^k)$. 
Note that now we cannot simply use formula~(\ref{eq:DP-main}).
In fact, if $|\partial(T')| < k$, formula~(\ref{eq:DP-main}) is well defined since each connected
component $T''$ of $T' -r$ has boundary of size at most $|\partial(T')| + 1 \leq k$ for any choice of $r$ (since $\partial(T'') \subseteq \partial(T') \cup \{r\}$).
However, if $|\partial(T')| =k$, it is possible that $|\partial(T'')| = k+1$, and formula~(\ref{eq:DP-main}) cannot be used.
Accordingly, there is no clear way to find the optimal vertex $r$. Instead, we choose a vertex $r_0$ such that
for every connected component $T''$ of $T'-r_0$, we have $|\partial(T'')|\leq k/2+1$. To prove that such a vertex exists,
we consider the tree $T'$ with vertex weights $w(u) = \left|\{v\in \partial(T'): (u,v)\in E\} \right|$ and find a balanced separator vertex
$r_0$ for $T'$ w.r.t. weights $w(u)$
(see Definition~\ref{def:bal-sep-vertex}). Then, the weight $w$ of every connected component $T''$ of $T'-r_0$
is at most $k/2$. Thus, $|\partial(T'')| \leq k/2+1 < 3k/4 <k$ (we add 1 because $r_0\in \partial(T'')$).

The above description implies that the only cases where our algorithm does not perform ``optimally" are the subproblems $T'$ with $|\partial(T')| = k$. It is also clear that these subproblems cannot occur too often, and more precisely, we can have at most 1 every $k / 2$ steps during the recursive decomposition into subproblems. Thus, we will distribute the cost (amortization) of each such non-optimal step that the algorithm makes over the previous $k/4$ steps before it occurs, whenever it occurs, and then show that all subproblems with boundary of size at most $3k/4$ are solved ``almost" optimally (more precisely, the solution to such a subproblem is $(1 + 4 / k)$-approximately optimal). This implies that the final solution will also be $(1 + 4/k)$-approximately optimal, since its boundary size is 0.

We now describe our algorithm in more detail. We keep two tables $B[T']$ and $C[T']$.
We will define their values so that we can find, using dynamic programming, a hub labeling for $T'$ of cost at most $B[T'] + C[T']$.
Informally, the table $C$ can be viewed as some extra budget that we use in order to pay for all the recursive steps with $|\partial(T')| = k$.
We let $C[T'] = \max \left\{0, \left(|\partial(T')| - 3k/4 \right) \cdot 4|T'| / k \right\}$ for every $T'$ with $|\partial(T')| \leq k$. We define $B$ (for $|T'|\geq 3$) by the following recurrence (where $r_0$ is a balanced separator):
\begin{align*}
B[T'] &= (1 + 4 / k) \cdot|T'|+ \min_{r \in T'} \sum_{T'' \text{ is a connected component of } T'-r} B[T''] \;\;, &&\text{ if } |\partial(T')| < k,\\
B[T'] &= \sum_{T'' \text{ is a connected component of } T'-r_0} B[T''] \;\;, &&\text{ if } |\partial(T')| = k.
\end{align*}
The base cases of our recursive formulas are when the subtree $T'$ is of size 1 or 2. In this case, we simply set $B[T'] = 1$, if $|T'| = 1$, and $B[T'] = 3$, if $|T'| = 2$.

In order to fill in the table, we generate all possible subsets of size at most $k$ that are candidate boundary sets, and for each such set we look at the resulting subtree, if any, and compute its size. We process subtrees in increasing order of size, which can be easily done if the generated subtrees are kept in buckets according to their size. Overall, the running time will be $n^{O(k)}$.

We will now show that the algorithm has approximation factor $(1 + 4/k)$ for any $k = 4t$, $t \geq 1$. That is, it is a polynomial-time approximation scheme (PTAS).

\begin{theorem}
The algorithm is a polynomial-time approximation scheme (PTAS) for \HL1 on trees.
\end{theorem}
\begin{proof}
We first argue about the approximation guarantee. The argument consists of an induction that incorporates the amortized analysis that was described above. More specifically, we will show that for any subtree $T'$, with $|\partial(T')| \leq k$, the total cost of the algorithm's solution is at most $B[T'] + C[T']$, and $B[T'] \leq \left(1 + \frac{4}{k} \right) \cdot OPT_{T'}$. Then, the total cost of the solution to the original HL instance is at most $B[T] + C[T]$, and, since $C[T] = 0$, we get that the cost is at most $(1 + 4/k) \cdot OPT$.

The induction is on the size of the subtree $T'$. For $|T'| = 1$ or $|T'| = 2$, the hypothesis holds. Let's assume now that it holds for all trees of size at most $t \geq 2$. We will argue that it then holds for trees $T'$ of size $t + 1$. We distinguish between the cases where $|\partial(T')| < k$ and $|\partial(T')| = k$:

\medskip
\noindent\textbf{Case ${|\partial(T')| < k}$:} Let $u_0 \in T'$ be the vertex that the algorithm picks and removes. The vertex $u_0$ is the minimizer of the expression $\min_{r' \in T'} \sum_{T'' \text{ is a connected component of } T'-r'} B[T'']$, and thus, using the induction hypothesis, we get that the total cost of the solution returned by the algorithm is at most:

\begin{equation*}
\begin{split}
    ALG(T') &\leq |T'| + \sum_{T'' \textrm{ is comp. of }T' - u_0} \Big( B[T''] + C[T''] \Big) \\
            &\leq |T'| + \sum_{T'' \textrm{ is comp. of }T' - u_0} B[T''] + \sum_{T'' \textrm{ is comp. of }T' - u_0} \max \left\{0, (|\partial(T')| + 1 - 3k/4) \cdot 4|T''|/k \right\} \\
            &= |T'| + \sum_{T'' \textrm{ is comp. of }T' - u_0} B[T''] + \max\{0, |\partial(T')| + 1 - 3k/4\} \cdot (4/k) \cdot \sum_{T'' \textrm{ is comp. of }T' - u_0}|T''| \\
            &\leq |T'| + \sum_{T'' \textrm{ is comp. of }T' - u_0} B[T''] + 4|T'|/k + \max\{0, |\partial(T')| - 3k/4 \} \cdot 4|T'| / k\\
            &\leq (1 + 4/k) \cdot |T'| + \left(\sum_{T'' \textrm{ is comp. of }T' - u_0} B[T''] \right) + C[T'] \\
            &= B[T'] + C[T'].
\end{split}
\end{equation*}

We proved the first part. We now have to show that $B[T'] \leq (1 + 4/k) OPT_{T'}$. Consider an optimal HL for $T'$. We may assume that it is a hierarchical labeling. Let $r \in T'$ be the vertex with the highest rank in this optimal solution.
We have, $OPT_{T'} = |T'| + \sum_{T'' \textrm{ is comp. of }T' - r} OPT_{T''}$. By definition, we have that
\begin{equation*}
\begin{split}
    B[T'] &= (1 + 4/k)|T'| + \min_{u \in T'} \sum_{T'' \textrm{ is comp. of }T' -u} B[T''] \\
          &\leq (1 + 4/k)|T'| + \sum_{T'' \textrm{ is comp. of }T' -r} B[T''] \\
          &\stackrel{(ind.hyp.)}{\leq} (1 + 4/k) |T'| + (1 + 4/k) \cdot \sum_{T'' \textrm{ is comp. of }T' -r} OPT_{T''} \\
          &= (1 + 4/k) \cdot OPT_{T'}.
\end{split}
\end{equation*}

\smallskip
\noindent\textbf{Case $|\partial(T')| = k$:}
 Using the induction hypothesis, we get that the total cost of the solution returned by the algorithm is at most:
\begin{equation*}
    ALG(T') \leq |T'| + \sum_{T'' \textrm{ is comp. of }T' -r_0} B[T''] + \sum_{T'' \textrm{ is comp. of }T' - r_0} C[T''].
\end{equation*}
By our choice of $r_0$, we have $|\partial(T'')| \leq 3k/4$, and so $C[T''] = 0$, for all trees $T''$ of the forest $T' - {r_0}$. Thus,
\begin{equation*}
\begin{split}
    ALG(T') &\leq |T'| + \sum_{T'' \textrm{ is comp. of }T' - {r_0}} B[T''] \\
            &= C[T'] + B[T'].
\end{split}
\end{equation*}
We now need to prove that $B[T'] \leq (1 + 4/k) \cdot OPT_{T'}$. We have,
\begin{equation*}
    B[T'] = \sum_{T'' \textrm{ is comp. of }T'- {r_0}} B [T'']
           \stackrel{(ind.hyp.)}{\leq} \sum_{T'' \textrm{ is comp of }T' -r_0} \left(1 + \frac{4}{k} \right) OPT_{T''}
           \leq \left(1 + \frac{4}{k} \right) OPT_{T'},
\end{equation*}
where in the the last inequality we use that $\sum_{T''} OPT_{T''}\leq OPT_{T'}$, which can be proved as follows.
We convert the optimal hub labeling $H'$ for $T'$ to a set of hub labelings for all subtrees $T''$ of $T'$:
the hub labeling $H''$ for $T''$ is the restriction of $H'$ to $T''$; namely, $H''_v = H'_v \cap V(T'')$ for every vertex $v\in T''$;
it is clear that the total number of hubs in labelings $H''$ for all subtrees $T''$ is at most the cost of $H'$.
Also, the cost of each hub labeling $H''$ is at least $OPT_{T''}$. The inequality follows.

We have considered both cases, $|S|<k$ and $|S|=k$, and thus shown that the hypothesis holds for any subtree $T'$ of $T$.
In particular, it holds for $T$. Therefore, the algorithm finds a solution of cost at most $B[T] + C[T] = B[T] \leq \left(1 + \frac{4}{k} \right) OPT$.

Setting $k = 4 \cdot \lceil 1 / \varepsilon \rceil$, as already mentioned, we get a $(1 + \varepsilon)$-approximation, for any fixed $\varepsilon \in (0,1)$, and the running time of the algorithm is $n^{O(1 / \varepsilon)}$.
\end{proof}


\subsection{A quasi-polynomial time algorithm for HL on trees}\label{quasi_sec}

The dynamic programming approach of the previous section can also be used to obtain a quasi-polynomial time algorithm for \HLp on trees. This is done by observing that the set of boundary vertices of a subtree is a subset of the hub set of every vertex in that subtree, and by proving that in an optimal solution for \HLp, all vertices have hub sets of size at most $O(\log n)$, for $p \in \{1, \infty\}$, and of size at most $O(\log^2 n)$ in the case of $p \in (1, \infty)$. Thus, restricting our DP to subinstances with polylogarithmic boundary size, we are able to get an exact quasi-polynomial time algorithm. Again, we only present the case for \HL1. More details for \HLp are given in Appendix~\ref{appendix_quasi_lp} (for \HLp, for every $p \in [1, \infty)$) and Appendix~\ref{appendix_quasi_infinity} (for \HL\infty).

In order to get an exact quasi-polynomial time algorithm for \HL1 on trees, we first prove that any optimal hierarchical hub labeling solution $H$ satisfies $\max_{u \in V} |H_u| = O(\log n)$.

\begin{theorem}\label{thm:quasi}
Let $H$ be an optimal (for \HL1) hierarchical hub labeling for a tree $T$. Then $\max_{u \in T} |H_u| = O(\log n)$.
\end{theorem}

We will need to prove a few auxiliary claims before we proceed with the proof of Theorem~\ref{thm:quasi}.
Consider a canonical hierarchical labeling for $T$. Recall that we may assume that $H$ is constructed as follows:
we choose a vertex $r$ in $T$ of highest rank and add it to each hub set $H_v$ (for $v$ in $T$), then recursively process each tree in the forest $T-r$.
For every vertex $u$, let $T_u$ be the subtree of $T$, which we process when we choose $r=u$.
(Note that $u$ is the vertex of highest rank in $T_u$.)

Let us consider the following ``move ahead" operation that transforms a hierarchical hub labeling $H$ to a hierarchical hub labeling $H'$.
The operation is defined by two elements $r_1 \prec r_2$ (i.e. $r_2 \in T_{r_1}$) as follows. Consider the process that constructs sets $H_u$,
described in the previous paragraph. Consider the recursive step when we process subtree $T_{r_1}$.
At this step, let us intervene in the process: choose vertex $r_2$ instead of choosing $r_1$. Then, when we recursively
process a subtree $T'$ of $T_{r_1} - r_2$, we choose the vertex of highest rank in $T'$ (as indicated by $H$). Clearly, we obtain a hierarchical hub labeling. Denote it by $H'$.

\begin{lemma}\label{lem:swap}
Suppose that we move $r_2$ ahead of $r_1$ ($r_1 \prec r_2$) and obtain labeling $H'$.
Let $\widetilde T$ be the connected component of $T_{r_1}- r_2$ that contains $r_1$.
The following statements hold.
\begin{enumerate}
\itemsep0em
\item If $u\notin T_{r_1}$, then $H'_u = H_u$.
\item If $u \in T_{r_1}$, then $H'_u \subset H_u \cup \{r_2\}$.
\item If $u\in T_{r_1} \setminus \widetilde{T}$, then $r_1\notin H'_u$.
\end{enumerate}
\end{lemma}
\begin{proof}
1. The ``move ahead" operation affects only the processing of the tree $T_{r_1}$ and its subtrees. Therefore, if $u\notin T_{r_1}$, then $H'_u = H_u$.

2. Consider a vertex $v\in H_u'\setminus \{r_2\}$; we will prove that $v\in H_u$. Consider the step when we add $v$ to $H_u'$. Denote the tree that we process at this step by $T'$. Then $u\in T'$ and $v$ is the highest ranked element in $T'$. In particular, $v\prec u$, and, by the definition of the ordering $\preceq$, $v\in H_u$.

3. If $u=r_2$, the statement is trivial. Otherwise, note that the subtree of $T_{r_1} - r_2$ that contains $u$ does not contain $r_1$ (because $r_1$ lies in the subtree $\widetilde T$),
and, thus, $u$ and $r_1$ never lie in the same subtree during later recursive calls.
Therefore, we do not add $r_1$ to $H'_u$.
\end{proof}


Consider a canonical hierarchical labeling $H$, a vertex $u$ and subtree $T_u$, as defined before. Let $T'$ be the largest connected component of $T_u - u$.
We define $M_u = |T_u| -|T'|$.

\begin{claim}\label{quasi-claim1}
Assume that $H$ is an optimal hierarchical hub labeling for a tree $T$.
Consider two vertices $u$ and $v$. If $u\prec v$, then  $M_u \geq M_v$.
\end{claim}
\begin{proof}
Without loss of generality, we can assume that there is no element $w$ between $u$ and $v$ in the partial order $\preceq$
(the result for arbitrary  $u\prec v$ will follow by transitivity).
Let $T'$ be the largest connected component of $T_u - u$.
Consider two cases.

\smallskip
\noindent \textbf{Case $v\notin T'$.} Then $M_v \leq |T_v| < |T_u| - |T'|$, since $T_v$ is a proper subset of $T_u\setminus T'$.

\smallskip
 \noindent \textbf{Case $v\in T'$.} Then $T'= T_v$, and $v$ is the highest ranked vertex in $T'$.
 Let $T''$ be the union of $v$ and all subtrees of $T_u-v$ that do not contain $u$.
 Note that $T'-T''$ is a connected component of $T' - v = T_v - v$, and, thus, $M_v \leq |T'| - |T' - T''| \leq |T''|$. Also, note that $v\in H_w$ for $w\in T'$, since $v$ is the highest ranked vertex in $T'$.

 Let us move $v$ ahead of $u$, and obtain a hub labeling $H'$.
 We have, by Lemma~\ref{lem:swap},
 \begin{enumerate}
 \item If $w\notin T_u$, then $H_w' = H_w$; thus, $|H_w'| = |H_w|$.
 \item If $w \in T_u - T'$, then $H_w' \subset  H_w \cup \{v\}$; thus, $|H_w'| \leq |H_w|+1$.
 \item If $w \in T' - T''$, then $H_w' \subset  H_w \cup \{v\} = H_w$; thus, $|H_w'| \leq |H_w|$.
 \item If $w \in T''$, then $H_w' \subset H_w$ and $u\in H_w\setminus H_w'$; thus, $|H_w'| \leq |H_w|-1$.
 \end{enumerate}
Since $H$ is an optimal labeling
 $$\sum_{w} |H_w|\leq \sum_{w} |H_w'|.$$
 Therefore, $|T_u - T'| \geq |T''|$. Recall that $M_u =|T_u - T'|$ and $|T''| \geq M_v$. We conclude that $M_u \geq M_v$.
 \end{proof}

\begin{claim}\label{quasi-claim2}
Assume that $H$ is an optimal hierarchical hub labeling for a tree $T$.
For every vertex $u$, we have $M_u \geq |T_u|/6$.
\end{claim}
\begin{proof}
Assume to the contrary that there is a vertex $u_1$ such that $M_{u_1} < |T_{u_1}|/6$. Denote $n' = |T_{u_1}|$.
Let $T'$ be the largest connected component of $T_{u_1} -u_1$,
and $u_2$ be the vertex of highest rank in $T'$. Let $T''$ be the largest connected component of $T' - u_2$, and $u_3$ be the vertex of highest rank in $T''$. Denote $M_1 = M_{u_1}$ and $M_2 = M_{u_2}$.
We have $u_1\preceq u_2$, and, therefore, $n'/6 > M_1\geq M_2$.
By the definition of $T'$, $T''$ and numbers $M_1$, $M_2$, we have
$$|T'| = |T_{u_1}| - M_1 > \frac{5}{6} n';\qquad  |T''| = |T'| - M_2 > \frac{4}{6} n'.$$
Note that $u_1$ belongs to $n'$ hub sets $H_w$ (namely, it belongs to $H_w$ for $w\in T_{u_1}$),
$u_2$ belongs to more than $\frac{5}{6}n'$ hub sets $H_w$, and $u_3$ belongs to more than $\frac{4}{6}n'$ hub sets $H_w$.

Let $c$ be balanced separator vertex in $T_{u_1}$. Clearly, in $H$ we have $u_1 \preceq c$. Consider the hub labeling $H'$ obtained by moving $c$ ahead of $u_1$. By Lemma~\ref{lem:swap}, $H'_w \subset H_w\cup \{c\}$ for $w\in T_{u_1}$, and $H'_w = H_w$ for $w\notin T_{u_1}$.
Since every connected component of $T_{u_1} -c $ contains at most $n'/2$ vertices, each of the vertices $u_1$, $u_2$, and $u_3$
will lie in a connected component of size at most $n'/2$. This means that $u_1$ belongs to at most $n'/2$ hub sets $H_w'$, and similarly $u_2$ and $u_3$ belong to at most $n'/2$ hub sets (each) in $H'$. We get that
$$\sum_w |H_w'| \leq \bigl(\sum_w |H_w|\bigr) + \underbrace{n' - 1}_{\substack{\text{bound on cost} \\ \text{of new hub } c}} - \Big(\underbrace{(n'-n'/2) + (5n'/6 - n'/2) + (4n'/6-n'/2)}_{
\substack{\text{difference in number of appearances} \\ \text{of hubs }  u_1, u_2, u_3 \text{ in hub labelings } H \text{ and } H'}}\Big) < \sum_w |H_w|.$$
We get a contradiction with the optimality of the hub labeling $\{H_w\}$.
\end{proof}

\begin{proof}[Proof of Theorem~\ref{thm:quasi}]
By Claim~\ref{quasi-claim2}, we get that during the decomposition of $T$ into subproblems, the sizes of the subproblems drop by a constant factor.
Therefore, the depth of the decomposition tree is $O(\log n)$, and so for every $u \in V$, we must have $|H_u| = O(\log n)$.
\end{proof}

We are now ready to present our quasi-polynomial time algorithm.

\begin{theorem}
There exists a DP algorithm that is an exact quasi-polynomial time algorithm for \HL1 on trees, with running time $n^{O(\log n)}$.
\end{theorem}
\begin{proof}
The DP algorithm is the same as the algorithm presented in Section \ref{ptas_trees}, with $k = O(\log n)$, without the balancing steps that occur when the boundary of a subproblem is of size exactly $k$.
The proof relies on the simple observation that the set of boundary vertices that describes a subtree is equal to, or a subset of, the set of hubs that have been assigned so far by an HHL solution to all the vertices of this subtree. In other words, the set of boundary vertices corresponds to some of the vertices of the path in the recursion tree that starts from the root (the vertex with the highest rank in $T$), and goes down to the internal vertex in the recursion tree that describes the current subproblem. By Theorem \ref{thm:quasi}, we know that any such path has at most $k= c \cdot \log n$ vertices, for some constant $c$. Thus, when looking for an optimal solution, we only have to consider subproblems that can be described by such paths, and so we can store entries in the dynamic program table only for subtrees $T'$ with $|\partial(T')| \leq k = c \cdot \log n$.

This implies that we can use a simple DP algorithm with entries defined by formula~(\ref{eq:DP-main}) of Section \ref{ptas_trees} for $T'$ with $|\partial(T')|\leq k$ and
$B[T']=\infty$ for $T'$ with $|\partial(T')|> k$.
As we already showed, the DP table has size $n^{O(k)} = n^{O(\log n)}$ (since we don't need to have any entries $B[T']$ for $T'$ with $|\partial(T')|> k$). The running time of the algorithm is $n^{O(\log n)}$.
\end{proof}

\section{Hardness of approximating hub labeling on general graphs}\label{Hardness}

In this section, we prove that \HL1 and \HL\infty are \NP-hard to approximate on general graphs with multiple shortest paths within a factor better than $\Omega(\log n)$, by using the $\Omega(\log n)$-hardness results for Set Cover. This implies that the current known algorithms for \HL1 and \HL\infty are optimal (up to constant factors).

\subsection{An \texorpdfstring{$\Omega(\log n)$}{Omega(log n)}-hardness for \texorpdfstring{\HL1}{HL1}}\label{appendix_l1_hardness}

In this section, we show that it is \NP-hard to approximate \HL1 on general graphs with multiple shortest paths within a factor better than $\Omega(\log n)$. We will use the hardness results for Set Cover, that, through a series of works spanning more than 20 years \cite{DBLP:journals/jacm/LundY94, DBLP:journals/jacm/Feige98, DBLP:conf/stoc/RazS97, DBLP:journals/talg/AlonMS06}, culminated in the following theorem.
\begin{theorem}[Dinur \& Steurer \cite{DBLP:conf/stoc/DinurS14}]\label{Set-Cover-Hardness}
For every $\alpha > 0$,  it is \NP-hard to approximate Set Cover to within a factor $(1 - \alpha) \cdot \ln n$, where $n$ is the size of the universe.
\end{theorem}

We start with an arbitrary unweighted instance of Set Cover. Let $\mathcal{X} = \{x_1, ..., x_n\}$ be the universe and $\mathcal{S} = \{S_1, ..., S_m\}$ be the family of subsets of $\mathcal{X}$, with $m = \poly(n)$. Our goal is to pick the smallest set of indices $I \subseteq [m]$ (i.e. minimize $|I|$) such that $\bigcup_{i \in I} S_i = \mathcal{X}$.

We will construct an instance of HL such that, given an $f(n)$-approximation algorithm for \HL1, we can use it to construct a solution for the original Set Cover instance of cost $O(f(\poly(n))) \cdot OPT_{SC}$, where $OPT_{SC}$ is the cost of the optimal Set Cover solution. Formally, we prove the following theorem.

\begin{theorem}\label{HL1_hardness_theorem}
Given an arbitrary unweighted Set Cover instance ($\mathcal{X}, \mathcal{S}$), $|\mathcal{X}|=n$, $|\mathcal{S}| = m$, with optimal value $OPT_{SC}$, and an $f(n)$-approximation algorithm for \HL1, there is an algorithm that returns a solution for the Set Cover instance of cost $O(f(\Theta(\max\{n, m\}^3))) \cdot OPT_{SC}$.
\end{theorem}

Using the above theorem, if we assume that $f(n) = o(\log n)$, then $O(f(\Theta(\max\{n,m\}^3))) = O(f(\mathtt{poly}(n))) = o(\log \mathtt{poly} (n)) = o(\log n)$, and so this would imply that we can get a $o(\log n)$-approximation algorithm for Set Cover. By Theorem \ref{Set-Cover-Hardness}, this is \NP-hard, and so we must have $f(n) = \Omega(\log n)$.

\begin{corollary}
It is \NP-hard to approximate \HL1 to within a factor $c \cdot \log n$, for some constant $c$, on general graphs (with multiple shortest paths).
\end{corollary}

Before proving Theorem \ref{HL1_hardness_theorem}, we need the following lemma.

\begin{lemma}\label{degree_one_vertices_lemma}
Let $G = (V, E)$ and $l: E \to \mathbb{R}^+$  be an instance of \HL1, in which there exists at least one vertex $u \in V$ with degree 1, and let $w$ be its unique neighbor. Then, any feasible solution $\{H_v\}_{v \in V}$  can be converted to a solution $H'$ of at most the same cost, with the property that $H_u' = H_w' \cup \{u\}$ and $u \notin H_v'$, for every vertex $v \neq u$.
\end{lemma}
\begin{proof}
Let $\{H_v\}$ be any feasible solution for HL, and let $u$ be a vertex of degree 1, and $w$ its unique neighbor. If the desired property already holds for every vertex of degree 1, then we are done. So let us assume that the property does not hold for some vertex
 $u$. Let $w$ be its neighbor and
$$B = \begin{cases}
         H_w \setminus \{u\}, & \mathrm{ if } \;\;|H_w \setminus \{u\}| \leq |H_u \setminus \{u\}|,\\
         H_u \setminus \{u\}, & \textrm{otherwise}.
   \end{cases}$$
We now set
\begin{itemize}
    \item $H_u' = B \cup \{u,w\}$.
    \item $H_w' = B \cup \{w\}$.
    \item $\forall v \in V \setminus \{u,w\}, \;H_v' = \begin{cases}
            H_v, & \mathrm{ if } \;u \notin H_v,\\
            (H_v \setminus \{u\}) \cup \{w\}, & \textrm{otherwise}.
   \end{cases}$
\end{itemize}
We first check the feasibility of $H'$. The pairs $\{u,w\}$, and $\{v,v\}$, for all $v \in V$, are clearly satisfied. Also, every pair $\{v,v'\}$ with $v,v' \notin \{u,w\}$ is satisfied, since $u \notin S_{vv'}$. Consider now a pair $\{u,v\}$, with $v \in V \setminus \{u,w\}$. If $u \in H_u \cap H_v$, we have $w \in H_u' \cap H_v'$. Otherwise, $\{u,v\}$ is covered with some vertex $z\in S_{uv} \setminus \{u\}$, and since $S_{uv} \setminus \{u\} = S_{wv}$, we have that $z \in H_u \cap H_v$ and $z \in H_w \cap H_v$.
It follows that $z\in H'_u \cap H'_v$. Now, consider a pair $\{w,v\}$, $v \in V \setminus \{u,w\}$. We have either $H_w' = (H_w \setminus \{u\}) \cup \{w\}$, which gives $H_w' \cap S_{wv} = H_w \cap S_{vw}$, or $H_w' = (H_u \setminus \{u\}) \cup \{w\}$. In the latter case, either $u \in H_v$ and so $w \in H_v'$, or $H_u \cap S_{uv} = H_u \cap S_{wv}$. It is easy to see that in all cases the covering property is satisfied.

We now argue about the cost of $H'$. We distinguish between the two possible values of $B$:
\begin{itemize}
    \setlength{\parskip}{0pt}
    \item $B = H_w \setminus \{u\}$: In this case, $|H_w'| \leq |H_w|$, since $w \in B$. If $u \in H_w$, then $|H_u'| = |H_w| \leq |H_u|$. Otherwise, it holds that $|H_w| \leq |H_u| - 1$, and so $|H_u'| = |H_w| + 1 \leq |H_u|$. For all $v \in V \setminus \{u,w\}$, it is obvious that $|H_v'| \leq |H_v|$.
    \item $B = H_u \setminus \{u\}$: If $w \in H_u$, then $|H_w'| = |H_u| - 1 < |H_w \setminus \{u\}| \leq |H_w|$, and $|H_u'| = |H_u|$. Otherwise, we must have $u \in H_w$, which means that $|H_u| < |H_w|$. Thus, $|H_w'| = |H_u| <  |H_w|$, and $|H_u'| = |H_u| + 1$, which gives $|H_u'| + |H_w'| < |H_u| + 1 + |H_w|$, and so $|H_u'| + |H_w'| \leq |H_u| + |H_w|$. Again, it is obvious that $|H_v'| \leq |H_v|$, for all $v \in V \setminus \{u,w\}$.
\end{itemize}
Thus, in all cases, $H'$ is a feasible hub labeling that satisfies the desired property and whose cost is $\sum_{v \in V} |H_v'| \leq \sum_{v \in V} |H_v|$.
\end{proof}

We are now ready to prove Theorem \ref{HL1_hardness_theorem}.
\begin{proof}[Proof of Theorem \ref{HL1_hardness_theorem}]
Given an unweighted Set Cover instance, we create a graph $G = (V, E)$ and \HL1 instance. 
We fix two parameters $A$ and $B$ (whose values we specify later) and do the following (see Figure \ref{HL_1 fig}):
\begin{itemize}
    \setlength{\parskip}{0pt}
	\item The 2 layers directly corresponding to the Set Cover instance are the $3^{rd}$ and the $4^{th}$ layer. In the $3^{rd}$ layer we introduce one vertex for each set $S_i \in \mathcal{S}$, and in the $4^{th}$ layer we introduce one vertex for each element $x_j \in \mathcal{X}$. We then connect $x_j$ to $S_i$ if and only if $x_j \in S_i$.
    \item The $2^{nd}$ layer contains $A$ vertices in total ($\{r_1, ..., r_A\}$), where $r_i$ is connected to every vertex $S_j$ of the $3^{rd}$ layer.
    \item For each vertex $r_i$ we introduce $B$ new vertices in the $1^{st}$ layer, where each such vertex $t_j^{(i)}$ has degree one and is connected to $r_i$.
    \item Similarly, for each $x_i$ we introduce $B$ new vertices in the $5^{th}$ layer, where each such $y_j^{(i)}$ has degree one and is connected to $x_i$.
    \item Finally, we introduce a single vertex $W$ in the $6^{th}$ layer, which is connected to every $x_i$ of the $4^{th}$ layer.
\end{itemize}
We also assign lengths to the edges. The (black) edges $(W, x_i)$ have length $\varepsilon < 1 / 2$ for every $x_i \in \mathcal{X}$, while all other (brown) edges have length 1. We will show that by picking the parameters $A$ and $B$ appropriately, we can get an $O(f(\poly(n))$-approximation for the Set Cover instance, given an $f(n)$-approximation for \HL1.

\begin{figure}[h]
\begin{center}
\scalebox{1.05}{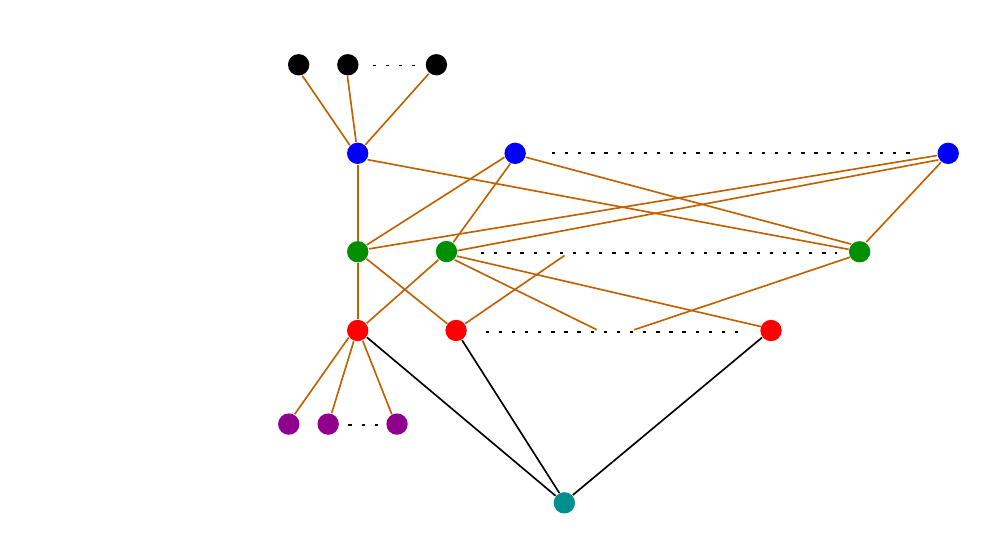}
\caption{The \HL1 instance corresponding to an arbitrary unweighted Set Cover instance \label{HL_1 fig}}
\end{center}
\end{figure}
We will now define a solution for this HL instance, whose cost depends on the cost of the optimal Set Cover. Let $I \subseteq [m]$ be the set of indices of an optimal Set Cover solution. We define the following HL solution, given in the array below. We use the notation $S(x_i)$ to denote an (arbitrarily chosen) set of the optimal Set Cover solution that covers $x_i$. We also write $[n] = \{1, ..., n\}$.
\[
\begin{array}{|c|c|}
\hline
    \mathbf{Layer} & \mathbf{Hubs}\\
    \hline
    1^{st} & \forall i \in [A], j \in [B], \;\;H_{t_j^{(i)}} = H_{r_i} \cup \{t_j^{(i)}\} \\ \hline
    2^{st} & \forall i \in [A], \;\; H_{r_i} = \{r_i\} \cup \{S_j: j \in I\} \\ \hline
    3^{st} & \forall i \in [m], \;\; H_{S_i} = \{S_i, W\} \cup \{r_1, ..., r_A \} \cup \{x_j: x_j \in S_i\} \\ \hline
    4^{st} & \forall i \in [n], \;\; H_{x_i} = \{x_i, W, S(x_i)\} \\ \hline
    5^{st} & \forall i \in [n], j \in [B], \;\; H_{y_j^{(i)}} = H_{x_i} \cup \{y_j^{(i)}\} \\ \hline
    6^{st} & H_W = \{W\} \cup \{S_1, ..., S_m\} \cup \{r_1, ..., r_A\} \\
    \hline
\end{array}
\]
We argue that the above solution is a feasible solution for the constructed instance. To this end, we consider all possible pairs of vertices for all layers. The notation ``$i$ - $j$" means that we check a pair with one vertex at layer $i$ and the other at layer $j$. The pairs $\{u,u\}$ are trivially satisfied, so we will not consider them below:
\begin{itemize}
    \setlength{\parskip}{0pt}
	\item 1 - 1: $\{t_a^{(i)}, t_b^{(j)}\}$. If $i = j$, then the pair is covered by $r_i$. If $i \neq j$, then $H_{t_a^{(i)}}$ and $H_{t_b^{(j)}}$ contain the same sets from the $3^{rd}$ layer, and so at least one shortest path is covered.
    \item 1 - 2: $\{t_a^{(i)}, r_j\}$. If $i = j$, then $r_i \in H_{t_a^{(i)}}$. If $i \neq j$, then $H_{t_a^{(i)}}$ and $H_{r_j}$ contain the same sets from the $3^{rd}$ layer, and so at least one shortest path is covered.
    \item 1 - 3: $\{t_a^{(i)}, S_j\}$. Observe that $r_i \in H_{S_j} \cap H_{t_a^{(i)}}$, and so the pair is covered.
    \item 1 - 4: $\{t_a^{(i)}, x_j\}$. The pair is covered by $S(x_j)$.
    \item 1 - 5: $\{t_a^{(i)}, y_b^{(j)}\}$. Again, the pair is covered by $S(x_j)$.
    \item 1 - 6: $\{t_a^{(i)}, W\}$. The pair is covered by $r_i$.
    \item 2 - 2: Any such pair is covered by any of the $S_j$'s, with $j \in I$.
    \item 2 - 3: $\{r_i, S_j\}$. We have $r_i \in H_{S_j}$.
    \item 2 - 4: $\{r_i, x_j\}$. The pair is covered by $S(x_j)$.
    \item 2 - 5: Since the ``2 - 4" pairs are covered, the ``2 - 5" pairs are covered as well.
    \item 2 - 6: Vertex $W$ is directly connected to all vertices of the $2^{nd}$ layer.
    \item 3 - 3: Any $\{S_i, S_j\}$ is covered by any vertex of the $2^{nd}$ layer.
    \item 3 - 4: $\{S_i, x_j\}$. If $x_j \in S_i$, then $x_j \in H_{S_{i}}$. If $x_j \notin S_i$, then $W \in H_{S_i} \cap H_{x_j}$.
    \item 3 - 5: $\{S_i, y_a^{(j)}\}$. If $x_j \in S_i$, then they are covered by $x_j$. If not, then they are covered by $W$.
    \item 3 - 6: There is a direct connection.
    \item 4 - 4: Any such pair is covered by $W$.
    \item 4 - 5: $\{x_i, y_a^{(j)}\}$. If $i = j$, then $x_i \in H_{y_a^{(i)}}$. If not, then they are covered by $W$.
    \item 4 - 6: There is a direct connection.
    \item 5 - 5: $\{y_a^{(i)}, y_b^{(j)} \}$. If $i = j$, then they are covered by $x_j$. If not, then they are covered by $W$.
    \item 5 - 6: There is a direct connection.
\end{itemize}

Thus, the above solution is indeed a feasible one. We compute its cost (each term from left to right corresponds to the total cost of the vertices of the corresponding layer):
\begin{equation*}
\begin{split}
	\textrm{COST} &\leq A \cdot B \cdot (OPT_{SC} + 2) + A \cdot (OPT_{SC} + 1) + m \cdot (A + n + 2) + 3n + 4B \cdot n + (1 + m + A) \\
                  &= O(AB \cdot OPT_{SC}) + O(A \cdot OPT_{SC}) + O(Am + mn) + O(n) + O(Bn) + O(m + A).
\end{split}
\end{equation*}
We set $A = B = \max\{m, n\}^{3/2}$, and let $N = \Theta(A \cdot B)$ denote the number of vertices of the graph. Then, the total cost is dominated by the term $AB \cdot OPT_{SC}$, and so we get that the cost $OPT$ of the optimal $HL_1$ solution is at most $c \cdot AB \cdot OPT_{SC}$, for some constant $c$. It is also easy to see that $OPT \geq A \cdot B$.

Assuming that we have an $f(n)$-approximation for \HL1, we can get a solution $H'$ of cost $\|H'\|_1 \leq c \cdot f(N) \cdot AB \cdot OPT_{SC}$. We will show that we can extract a feasible Set Cover solution of cost at most $\frac{\|H'\|_1}{AB}$.

To extract a feasible Set Cover, we first modify $H'$. Using Lemma \ref{degree_one_vertices_lemma}, we can obtain a solution $H''$ such that for every $t_i^{(j)}$ we have $H_{t_i^{(j)}}'' = H_{r_j}'' \cup \{t_i^{(j)}\}$ and $t_i^{(j)} \notin H_v''$, for every vertex $v \neq t_i^{(j)}$. We also apply the lemma for the vertices of the $5^{th}$ layer. In addition to that, we add $\{r_1, ..., r_A\}$ to the hub set of $W$, increasing the cost by at most $A$. Thus, we end up with a solution $H''$ of cost at most $c \cdot f(N) \cdot AB \cdot OPT_{SC} + A \leq c' \cdot f(N) \cdot AB \cdot OPT_{SC}$.

We now look at every vertex $r_i$ of the $2^{nd}$ layer for which we have $x_j \in H_{r_i}''$, for some $j \in [n[$. This $x_j$ can only be used to connect $r_i$ to $x_j$  and to the $\{y_a^{(j)}\}$'s. In that case, we can remove $x_j$ from $H_{r_i}''$ and all $H_{t_b^{(i)}}''$, and add $r_i$ to $H_{x_j}''$ and to all $H_{y_a^{(j)}}''$. The cost of the solution cannot increase, and we again call this new solution $H''$.

We are ready to define our Set Cover solution. For each $i \in [A]$, we define $F_i = H_{r_i}'' \cap \{S_1, ..., S_m\}$. Let $Z_i = |\mathcal{X} \setminus \bigcup_{j \in F_i} S_j|$ be the number of uncovered elements. If $Z_i = 0$, then $F_i$ is a valid Set Cover. If $Z_i > 0$, then we cover the remaining elements using some extra sets (at most $Z_i$ such sets). At the end, we return $\min_{i \in [A]} \{|F_i| + Z_i\} $.

In order to analyze the cost of the above algorithm, we need the following observation. Let us look at $H_{r_i}''$, and an element $x_j$ that is not covered. By the structure of $H''$, this means that $r_i \in H_{x_j}''$. Thus, the number of uncovered elements $Z_i$ contributes a term $(B + 1) Z_i$ to the cost of the \HL1 solution.  For each $i$, the number of uncovered elements $Z_i$ implies an increase in the cost of the $4^{th}$ and $5^{th}$ layer, which is ``disjoint" with the increase for $j \neq i$, and so the total cost of the $1^{st}$, $2^{nd}$, $4^{th}$ and $5^{th}$ layer combined is at least $\sum_{i = 1}^A (B + 1) (|H_{r_i}''| + Z_i)$. So, there must exist an $i$ such that
\begin{equation*}
	|H_{r_i}''| + Z_i \leq \frac{\|H''\|_1}{A(B + 1)}.
\end{equation*}
We pick the Set Cover with cost at most $\min_j \{|F_j| + Z_j\} \leq \min_j \{|H_{r_j}''| + Z_j\}$, and so we end up with a feasible Set Cover solution of cost at most
\begin{equation*}
	\frac{c' \cdot f(N) \cdot AB \cdot OPT_{SC}}{AB} = c' \cdot f(N) \cdot OPT_{SC}.
\end{equation*}
\end{proof}

\subsection{An \texorpdfstring{$\Omega(\log n)$}{Omega(log n)}-hardness for \texorpdfstring{\HL\infty}{HL-infinity}}\label{appendix_infinity_hardness}

In this section, we will show that it is \NP-hard to approximate \HL\infty to within a factor better than $\Omega(\log n)$. We will again use the hardness results for Set Cover.

\begin{theorem}
Given an arbitrary unweighted Set Cover instance $(\mathcal{X}, \mathcal{S})$, $|\mathcal{X}| = n$, $|\mathcal{S}| = m = \poly(n)$, with optimal value $OPT_{SC}$, and an $f(n)$-approximation algorithm for \HL\infty, there is an algorithm that returns a solution for the Set Cover instance with cost at most $O \left(f(O(n^4 \cdot m)) \right) \cdot OPT_{SC}$.
\end{theorem}
\begin{proof}
Let $\mathcal{X} = \{x_1, ..., x_n\}$ and $\mathcal{S} = \{S_1, ..., S_m\}$, $S_i \subseteq \mathcal{X}$ be a Set Cover instance. We will construct an instance of \HL\infty, such that, given an $f(n)$-approximation algorithm for it, we will be able to solve the Set Cover instance within a factor of $O \left(f(O(n^4m)) \right)$. We now describe our construction:
\begin{itemize}
    \setlength{\parskip}{0pt}
    \item We introduce a complete bipartite graph $(A,B,E)$. By slightly abusing notation, we denote $|A| = A$ and $|B| = B$, where $A$ and $B$ are two parameters to be set later on.
    \item Each vertex $u \in A$ ``contains'' $K$ vertices $\{r_{u,1}, ..., r_{u,K}\}$.
    \item Each vertex vertex $v \in B$ ``contains" a copy of the universe $\{x_{v,1}, ..., x_{v,n}\}$.
    \item Each edge $(u,v)$ is replaced by an intermediate layer of vertices $\mathcal{S}_{uv} = \{S_{uv,1}, ..., S_{uv,m}\}$, which is essentially one copy of $\mathcal{S}$. We then connect every vertex $r_{u,i}$, $i \in [K]$, to every vertex $S_{uv,j}$, $j \in [m]$, and we also connect each $S_{uv,j}$ to $x_{v,t}$, if $x_t \in S_j$. All these edges (colored red in the figure) have length 1.
    \item Finally, we introduce three extra vertices $W_A$ and $W_B$ and $W_S$, and the edges $(W_A, r_{u,i})$, for all $u\in A$ and $i \in [K]$, the edges $(W_B, x_{v,j})$, for all $v \in B$ and $j \in [n]$, and the edges $(W_S, S_{uv,j})$, for all $u \in A$, $v \in B$ and $j \in [m]$. All these edges (colored black in the figure) have length $\varepsilon < 1$.
\end{itemize}
The construction is summarized in Figures \ref{HL infty fig1}, \ref{HL infty fig2}.

\begin{figure}[ht!]
\centering
\begin{subfigure}{.5\textwidth}
  \centering
\scalebox{0.5}{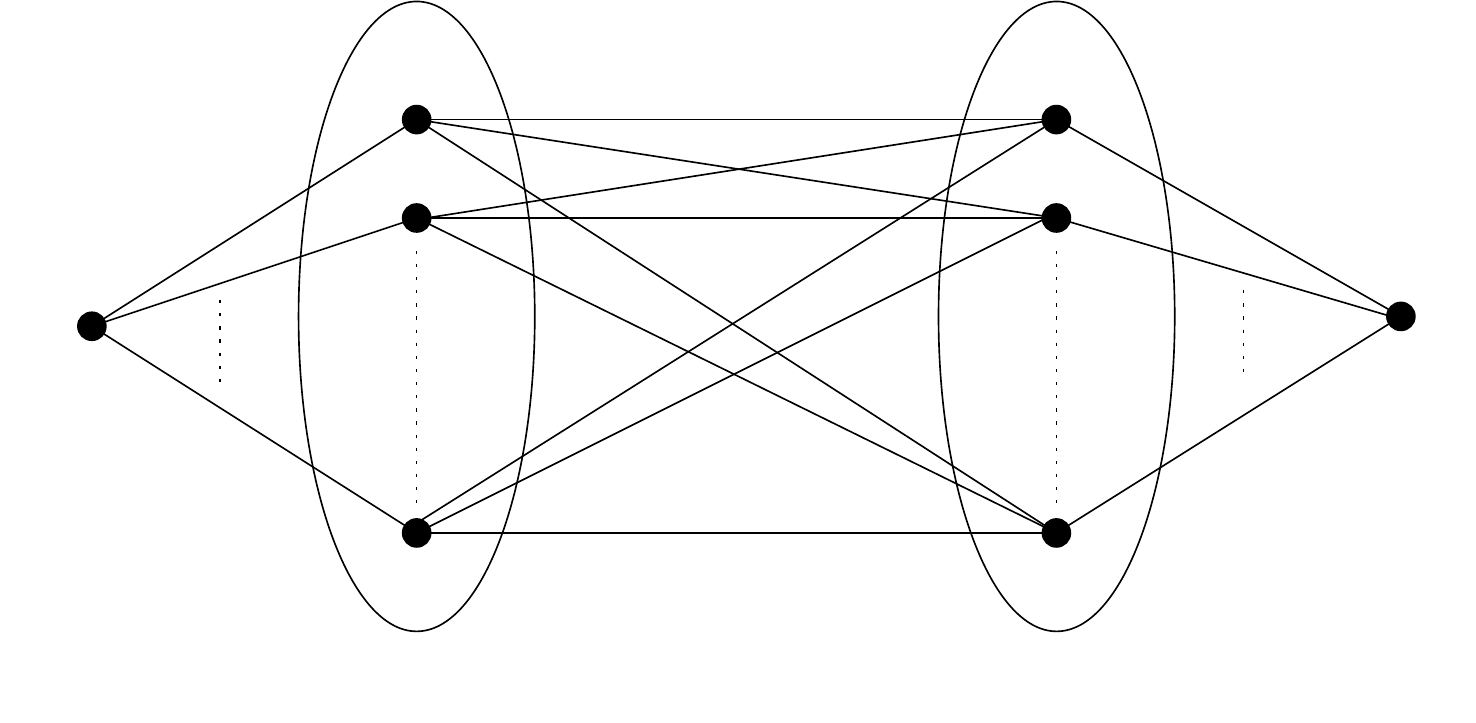}
  \caption{The general structure of the graph}
\label{HL infty fig1}
\end{subfigure}%
\begin{subfigure}{.5\textwidth}
  \centering
  \scalebox{0.4}{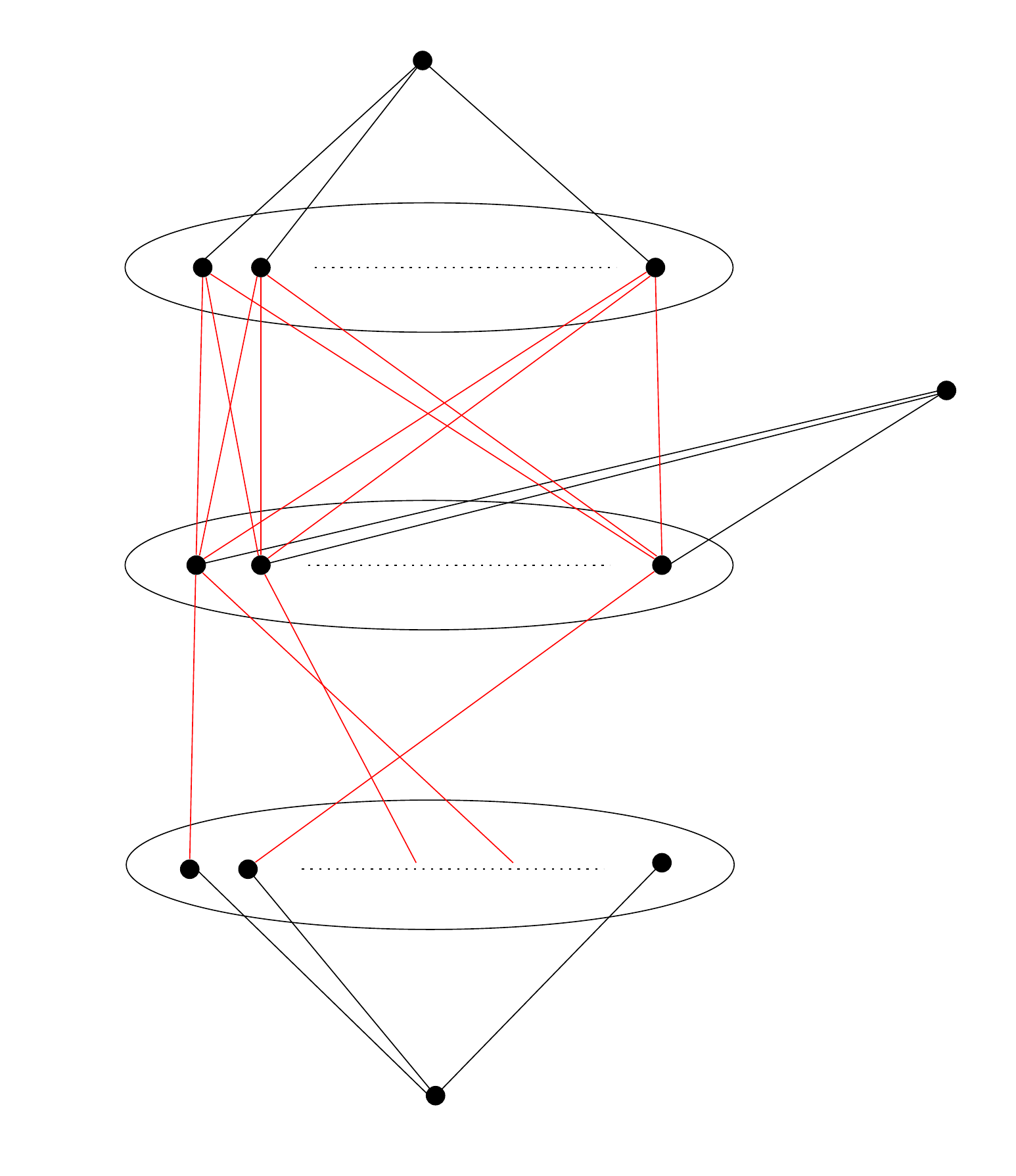}
  \caption{A closer look at an edge $(u,v) \in E$}
\label{HL infty fig2}
\end{subfigure}
\end{figure}

In the resulting construction, the number of vertices, denoted by $N$, is $N = AK + Bn + ABm + 3$, and the number of edges, denoted by $M$, is at least $M \geq AB(Km + m) + AK + ABm + Bn$. Let $OPT$ denote the cost of an optimal \HL\infty solution $H$ for this instance. Then, by a standard pigeonhole principle argument, and since every edge is a unique shortest path, we get that $OPT \geq \frac{M}{N}$. We now set the parameters, as follows: $A = B = K = n^2$. With these values, we have $N = \Theta(n^4 \cdot m)$, $M = \Omega(n^6 \cdot m)$ and $OPT = \Omega(n^2)$.

We will describe an intended feasible solution for this instance, that will give an upper bound on OPT. Let $I \subseteq[m]$ denote an optimal Set Cover of our original Set Cover instance, and let $I_j \in I$ denote the index of a chosen set that covers $x_j$. The HL solution is the following:
\begin{itemize}
    \setlength{\parskip}{0pt}
	\item $H_{r_{u,i}} = \{r_{u,i}\} \cup \left( \bigcup_{{v \in B}} \{ S_{uv,j}: j \in I \} \right) \cup \{W_A, W_B, W_S\}$, for $u \in A$ and $i \in [K]$.
    \item $H_{x_{v,j}} = \{x_{v,j}\} \cup \left( \bigcup_{u \in A} \{S_{uv,I_j}\} \right) \cup \{W_A, W_B, W_S\}$, for $v \in B$ and $j \in [n]$.
    \item $H_{S_{uv,t}} = \{ S_{uv,t} \} \cup \{r_{u,1}, ..., r_{u,K} \} \cup \{ x_{v,1}, ..., x_{v,n} \} \cup \{W_A, W_B, W_S\}$, for $u\in A$, $v \in B$ and $t \in [m]$.
    \item $H_{W_q} = \{W_A, W_B, W_S\}$, for $q \in \{A,B,S\}$.
\end{itemize}
We now compute the sizes of these hub sets. We have:
\begin{itemize}
    \setlength{\parskip}{0pt}
	\item $|H_{r_{u,i}}| = B |I| + 4 = \Theta(n^2 \cdot |I|)$, for $u \in A$ and $i \in [K]$.
    \item $|H_{x_{v,j}}| = A + 4 = \Theta(n^2)$, for $v \in B$ and $j \in [n]$.
    \item $|H_{S_{uv,t}}| = K + n + 4 = \Theta(n^2)$, for $u\in A$, $v \in B$ and $t \in [m]$.
    \item $|H_{W_q}| = 3$, for $q \in \{A,B,S\}$.
\end{itemize}
Thus, we get that the value of the above solution is $\|H\|_\infty = Val = \Theta(n^2 \cdot |I|)$. We now show that the above is indeed a feasible solution. For that, we consider all possible pairs of vertices:
\begin{itemize}
    \setlength{\parskip}{0pt}
	\item $r_{u,i}$ - $r_{v,j}$: $W_A$ is a hub for both vertices.
    \item $r_{u,i}$ - $S_{uv,j}$: $r_{u,i} \in H_{S_{uv,j}}$.
    \item $r_{u,i}$ - $S_{wv,j}$, $w \neq u$, $v \neq u$: $W_S$ is a hub for both vertices.
    \item $r_{u,i}$ - $x_{v,j}$: $S_{uv,I_j}$ is a hub for both vertices.
    \item $r_{u,i}$ - $W_q$, for $q \in \{A,B,S\}$: $W_q \in H_{r_{u,i}}$.
    \item $S_{uv,i}$ - $S_{u'v',j}$: $W_S$ is a hub for both vertices.
    \item $S_{uv,i}$ - $x_{v,j}$:  $x_{v,j} \in H_{S_{uv,i}}$.
    \item $S_{uv,i}$ - $x_{v',j}$, $u \neq v'$, $v \neq v'$: $W_B$ (or $W_S$) is a hub for both vertices.
    \item $S_{uv,i}$ - $W_q$, for $q \in \{A,B,S\}$: $W_q \in H_{S_{uv,i}}$.
    \item $x_{v,i}$ - $x_{v',j}$: $W_B$ is hub for both vertices.
    \item $x_{v,j}$ - $W_q$, for $q \in \{A,B,S\}$: $W_q \in H_{x_{v,j}}$
    \item $W_q$ - $W_{q'}$, for $q, q' \in \{A,B,S\}$: $W_{q'} \in H_{W_q}$.
\end{itemize}
Thus, the proposed solution is indeed a feasible solution. Assuming now that we have an $f(n)$-approximation algorithm for \HL\infty, we can obtain a solution $H'$ of cost $\|H'\|_\infty \leq f(N) \cdot OPT \leq c \cdot f(N) \cdot n^2 \cdot |I|$. We will now show that we can extract a feasible solution for the original Set Cover instance, of cost $O(f(N)) \cdot |I|$. Let $H'$ be the solution that the algorithm returns. As a reminder, we have already proved that $\|H'\|_\infty = \Omega(n^2)$. We first transform $H'$ to a solution $H''$ that will look more like our intended solution, as follows:
\begin{itemize}
    \setlength{\parskip}{0pt}
	\item $H''_{S_{uv},t} := H_{S_{uv},t}' \cup \{r_{u,1}, ..., r_{u,K} \} \cup \{ x_{v,1}, ..., x_{v,n} \} \cup \{W_A, W_B, W_S\}$, for $u\in A$, $v \in B$ and $t \in [m]$. We have $|H_{S_{uv},t}''| \leq |H_{S_{uv},t}'| + K + n + 3 \leq \|H'\|_\infty + n^2 + n + 3 = O(\|H'\|_\infty)$.
    \item $H_{W_q}'' := H_{W_q}' \cup \{W_A, W_B, W_S\}$, for $q \in \{A,B,S\}$. We have $|H_{W_q}''| \leq |H_{W_q}'| + 3 = O(|H_{W_q}'|) = O(\|H'\|_\infty)$.
    \item We now look at $H_{r_{u,i}}'$. For every $x_j \in \mathcal{X}$, we pick (arbitrarily) a set $S(x_j) \in \mathcal{S}$ with $x_j \in S(x_j)$, that we will use to cover it. Now, if $x_{v,j} \in H_{r_{u,i}}'$, we remove $x_{v,j}$ from $H_{r_{u,i}}'$ and add $S_{uv}(x_j)$ to $H_{r_{u,i}}'$. This doesn't change the size of $H_{r_{u,i}}'$. We also add $S_{uv}(x_j)$ to the hub set of $x_{v,j}$. This increases the size of $H_{x_{v,j}}'$ by 1. The crucial observation here is that since we have decided in advance which set we will use to cover $x_j$, then $|H_{x_{v,j}}'|$ can only increase by 1, for every edge $(u,v)$. Thus, the total increase in $|H_{x_{v,j}}'|$ is at most $A$, i.e. $|H_{x_{v,j}}''| \leq |H_{x_{v,j}}'| + n^2 = O(\|H'\|_\infty)$.
\end{itemize}
The above transformed solution, as shown, has the same (up to constant factors) cost as the solution that the algorithm returns, i.e. $\|H''\|_\infty = O(\|H'\|_\infty) = O(f(N)) \cdot n^2 \cdot |I|$, and is clearly feasible.

In order to recover a good Set Cover solution, we look at the sets $H_{r_{u,i}}'' \cap \mathcal{S}_{uv}$. Each such intersection can be viewed as a subset $C_{u,v,i}$ of $\mathcal{S}$. Let $Z_{u,v,i}$ denote the number of elements that are not covered by $C_{u,v,i}$, i.e. $Z_{u,v,i} = |\mathcal{X} \setminus (\bigcup_{S \in C_{u,v,i}} S)|$. Our goal is to show that there exists a $\{u,v,i\}$ such that $|C_{u,v,i}| + Z_{u,v,i} = O(\|H''\|_\infty / n^2)$. Since there is a polynomial number of choices of $\{u,v,i\}$, we can then enumerate over all choices and find a Set Cover with cost $O(f(N)) \cdot |I|$.

To prove that such a good choice exists, we will make a uniformly random choice over $\{u,v,i\}$, and look at the expected value $\E \left[|C_{u,v,i}| + Z_{u,v,i} \right]$. We have $\E\left[|C_{u,v,i}| + Z_{u,v,i}\right] = \E\left[|C_{u,v,i}|\right] + \E[Z_{u,v,i}]$. We look separately at the two terms. We make the following 2 observations:
\begin{equation*}
	\sum_{v \in B} |C_{u,v,i}| \leq |H_{r_{u,i}}''| = O(\|H'\|_\infty),
\end{equation*}
and
\begin{equation*}
	\sum_{u \in A} \sum_{i \in K} Z_{u,v,i} \leq \sum_{j \in [n]} |H_{x_{v,j}}''| = n \cdot O(\|H'\|_\infty).
\end{equation*}
The second observation follows from the fact that for any given $r_{u,i}$ and edge $(u,v)$, the uncovered elements $x_{v,j}$ must have $r_{u,i} \in H_{x_{v,j}}''$. With these, we have
\begin{equation*}
	\E[|C_{u,v,i}|] = \frac{1}{ABK} \sum_{u \in A, i \in [K]} \sum_{v \in B} |C_{u,v,i}| \leq \frac{1}{ABK} \cdot A K \cdot O(\|H'\|_\infty) = O(\|H'\|_\infty / B).
\end{equation*}
Similarly,
\begin{equation*}
	\E[Z_{u,v,i}] = \frac{1}{ABK} \sum_{v \in B} \sum_{u \in A} \sum_{i \in K} Z_{u,v,i} \leq \frac{1}{ABK} \cdot B \cdot n \cdot O(\|H'\|_\infty) = \frac{n}{AK} \cdot O(\|H'\|_\infty).
\end{equation*}
Thus, we get that $\E\left[|C_{u,v,i}| + Z_{u,v,i}\right] = \left(\frac{1}{B} + \frac{n}{AK} \right) \cdot O(\|H'\|_\infty) = O(\|H'\|_\infty / n^2) = O(f(N)) \cdot |I|$. This means that there exists a choice of $\{u,v,i\}$ such that the corresponding Set Cover has size $O(f(N)) \cdot |I|$. As already mentioned, there are polynomially many choices, so we can enumerate them and find the appropriate $\{u,v,i\}$, and, thus, recover a Set Cover solution for our original Set Cover instance of cost $O(f(N)) \cdot |I|$, where, as already stated, $N = \Theta(n^4 \cdot m)$.
\end{proof}

\begin{corollary}
It is \NP-hard to approximate \HL\infty to within a factor better than $\Omega(\log n)$.
\end{corollary}
\begin{proof}
 	The previous theorem gives an $O(f(O(n^4m))$-approximation algorithm for Set Cover, given that an $f(n)$-approximation algorithm for \HL\infty exists. If we assume that there exists such an algorithm with $f(n) = o(\log n)$, then we could use it to approximate Set Cover within a factor $o(\log O(n^4m)) = o(\log \poly(n)) = o(\log n)$, and this is impossible, assuming that $\P \neq \NP$.
\end{proof}

\vspace{10pt}
\section*{Acknowledgments}
We would like to thank Robert Krauthgamer and Konstantin Makarychev for useful discussions, and the anonymous referees for their valuable comments. Research was supported by the second author's NSF grants CAREER CCF-1150062 and IIS-1302662.
\vspace{10pt}

\bibliography{references}

\begin{thebibliography}{10}

\bibitem{DBLP:conf/icalp/AbrahamDFGW11}
Ittai Abraham, Daniel Delling, Amos Fiat, Andrew~V. Goldberg, and Renato
  Fonseca~F. Werneck.
\newblock {VC-Dimension and Shortest Path Algorithms}.
\newblock In {\em Proceedings of the 38th International Colloquium on Automata,
  Languages, and Programming (ICALP)}, pages 690--699, 2011.

\bibitem{DBLP:conf/gis/AbrahamDFGW12}
Ittai Abraham, Daniel Delling, Amos Fiat, Andrew~V. Goldberg, and Renato
  Fonseca~F. Werneck.
\newblock {HLDB:} location-based services in databases.
\newblock In {\em Proceedings of the 20th International Conference on Advances
  in Geographic Information Systems (SIGSPATIAL - formerly known as GIS)},
  pages 339--348, 2012.

\bibitem{DBLP:journals/jea/AbrahamDGW13}
Ittai Abraham, Daniel Delling, Andrew~V. Goldberg, and Renato~F. Werneck.
\newblock Alternative routes in road networks.
\newblock {\em {ACM} Journal of Experimental Algorithmics}, 18, 2013.

\bibitem{DBLP:conf/wea/AbrahamDGW11}
Ittai Abraham, Daniel Delling, Andrew~V. Goldberg, and Renato Fonseca~F.
  Werneck.
\newblock {A Hub-Based Labeling Algorithm for Shortest Paths in Road Networks}.
\newblock In {\em Proceedings of the 10th International Symposium on
  Experimental Algorithms (SEA)}, pages 230--241, 2011.

\bibitem{DBLP:conf/esa/AbrahamDGW12}
Ittai Abraham, Daniel Delling, Andrew~V. Goldberg, and Renato Fonseca~F.
  Werneck.
\newblock {Hierarchical Hub Labelings for Shortest Paths}.
\newblock In {\em Proceedings of the 20th Annual European Symposium on
  Algorithms (ESA)}, pages 24--35, 2012.

\bibitem{DBLP:conf/soda/AbrahamFGW10}
Ittai Abraham, Amos Fiat, Andrew~V. Goldberg, and Renato Fonseca~F. Werneck.
\newblock {Highway Dimension, Shortest Paths, and Provably Efficient
  Algorithms}.
\newblock In {\em Proceedings of the 21st Annual {ACM-SIAM} Symposium on
  Discrete Algorithms (SODA)}, pages 782--793, 2010.

\bibitem{DBLP:journals/talg/AlonMS06}
Noga Alon, Dana Moshkovitz, and Shmuel Safra.
\newblock Algorithmic construction of sets for \emph{k}-restrictions.
\newblock {\em {ACM} Trans. Algorithms}, 2(2):153--177, 2006.

\bibitem{DBLP:conf/icalp/BabenkoGGN13}
Maxim~A. Babenko, Andrew~V. Goldberg, Anupam Gupta, and Viswanath Nagarajan.
\newblock {Algorithms for Hub Label Optimization}.
\newblock In {\em Proceedings of the 40th International Colloquium on Automata,
  Languages, and Programming (ICALP)}, pages 69--80, 2013.

\bibitem{DBLP:conf/mfcs/BabenkoGKSW15}
Maxim~A. Babenko, Andrew~V. Goldberg, Haim Kaplan, Ruslan Savchenko, and
  Mathias Weller.
\newblock {On the Complexity of Hub Labeling}.
\newblock In {\em Proceedings of the 40th International Symposium on
  Mathematical Foundations of Computer Science 2015 (MFCS)}, pages 62--74,
  2015.

\bibitem{DBLP:journals/corr/BastDGMPSWW15}
Hannah Bast, Daniel Delling, Andrew~V. Goldberg, Matthias
  M{\"{u}}ller{-}Hannemann, Thomas Pajor, Peter Sanders, Dorothea Wagner, and
  Renato~F. Werneck.
\newblock Route planning in transportation networks.
\newblock {\em CoRR}, abs/1504.05140, 2015.

\bibitem{Bast06}
Holger Bast, Stefan Funke, and Domagoj Matijevic.
\newblock {TRANSIT - ultrafast shortest-path queries with linear-time
  preprocessing}.
\newblock In {\em 9th DIMACS Implementation Challenge}, 2006.

\bibitem{Berend_improvedbounds}
Daniel Berend and Tamir Tassa.
\newblock Improved bounds on bell numbers and on moments of sums of random
  variables.
\newblock {\em Probability and Math. Statistics}, 30:185--205, 2010.

\bibitem{DBLP:journals/tit/Breuer66}
Melvin~A. Breuer.
\newblock Coding the vertexes of a graph.
\newblock {\em {IEEE} Trans. Information Theory}, 12(2):148--153, 1966.

\bibitem{Breuer1967583}
Melvin~A Breuer and Jon Folkman.
\newblock An unexpected result in coding the vertices of a graph.
\newblock {\em Journal of Mathematical Analysis and Applications},
  20(3):583--600, 1967.

\bibitem{DBLP:journals/siamcomp/CohenHKZ03}
Edith Cohen, Eran Halperin, Haim Kaplan, and Uri Zwick.
\newblock {Reachability and Distance Queries via 2-Hop Labels}.
\newblock {\em {SIAM} J. Comput.}, 32(5):1338--1355, 2003.

\bibitem{DBLP:conf/stoc/DinurS14}
Irit Dinur and David Steurer.
\newblock Analytical approach to parallel repetition.
\newblock In {\em Proceedings of the 46th ACM Symposium on Theory of Computing
  (STOC)}, pages 624--633, 2014.

\bibitem{DBLP:journals/jacm/Feige98}
Uriel Feige.
\newblock A threshold of ln \emph{n} for approximating set cover.
\newblock {\em J. {ACM}}, 45(4):634--652, 1998.

\bibitem{DBLP:journals/jal/GavoillePPR04}
Cyril Gavoille, David Peleg, St{\'{e}}phane P{\'{e}}rennes, and Ran Raz.
\newblock Distance labeling in graphs.
\newblock {\em J. Algorithms}, 53(1):85--112, 2004.

\bibitem{Gilmore}
P.C. Gilmore.
\newblock Families of sets with faithful graph representations.
\newblock {\em {Res. Note N. C. (2nd ed.)}}, 184, 1962.

\bibitem{Graham1972}
R.~L. Graham and H.~O. Pollak.
\newblock On embedding graphs in squashed cubes.
\newblock {\em Lecture Notes in Mathematics}, 303:99--110, 1972.

\bibitem{gyarfas1970helly}
Andr{\'a}s Gy{\'a}rf{\'a}s and Jen{\"o} Lehel.
\newblock A {H}elly-type problem in trees.
\newblock {\em Colloq. Math. Soc. J. Bolyai.}, 4, 1970.

\bibitem{DBLP:journals/siamdm/KannanNR92}
Sampath Kannan, Moni Naor, and Steven Rudich.
\newblock {Implicit Representation of Graphs}.
\newblock {\em {SIAM} J. Discrete Math.}, 5(4):596--603, 1992.

\bibitem{DBLP:journals/jacm/LundY94}
Carsten Lund and Mihalis Yannakakis.
\newblock On the hardness of approximating minimization problems.
\newblock {\em J. {ACM}}, 41(5):960--981, 1994.

\bibitem{DBLP:journals/jgt/Peleg00}
David Peleg.
\newblock Proximity-preserving labeling schemes.
\newblock {\em Journal of Graph Theory}, 33(3):167--176, 2000.

\bibitem{DBLP:conf/stoc/RazS97}
Ran Raz and Shmuel Safra.
\newblock {A Sub-Constant Error-Probability Low-Degree Test, and a Sub-Constant
  Error-Probability PCP Characterization of NP}.
\newblock In {\em Proceedings of the 29th {ACM} Symposium on Theory of
  Computing (STOC)}, pages 475--484, 1997.

\bibitem{DBLP:conf/esa/White15}
Colin White.
\newblock Lower bounds in the preprocessing and query phases of routing
  algorithms.
\newblock In {\em Algorithms - {ESA} 2015 - 23rd Annual European Symposium,
  Patras, Greece, September 14-16, 2015, Proceedings}, pages 1013--1024, 2015.

\bibitem{DBLP:journals/combinatorica/Winkler83}
Peter~M. Winkler.
\newblock Proof of the squashed cube conjecture.
\newblock {\em Combinatorica}, 3(1):135--139, 1983.

\end{thebibliography}
\bibliographystyle{plain}

\appendix
\section{\texorpdfstring{\HLp}{HLp} on graphs with unique shortest paths}\label{sec:lp-norm-algorithm}

In this section, we analyze Algorithm~\ref{Pre-Hubs_Algorithm}, assuming that the objective function is $\left(\sum_{u \in V} |H_u|^p\right)^{1 / p}$, for arbitrary fixed $p \geq 1$. To do so, we first modify $\mathbf{LP_1}$ and turn it into a convex program, with the same constraints and the following objective function:
\begin{equation*}
    \min: \; \left(\sum_{u \in V} \left(\sum_{v \in V} x_{uv} \right)^p\right)^{1 / p}.
\end{equation*}
To analyze the performance of Algorithm \ref{Pre-Hubs_Algorithm} in this case, we need the following theorem by Berend and Tassa \cite{Berend_improvedbounds}.
\begin{theorem}[Theorem 2.4, \cite{Berend_improvedbounds}]\label{moment_theorem}
Let $X_1, ..., X_t$ be a sequence of independent random variables for which $\Prob[0 \leq X_i \leq 1] = 1$, and let $X = \sum_{i = 1}^t X_i$. Then, for all $p \geq 1$,
\begin{equation*}
    \left(\E[X^p]\right)^{1 / p} \leq 0.792 \cdot v(p) \cdot \frac{p}{\ln(p + 1)} \cdot \max\{\E[X]^{1 / p}, \E[X]\},
\end{equation*}
where
\begin{equation*}
    v(p) = \left(1 + \frac{1}{\lfloor p \rfloor} \right)^\frac{\{p\} \cdot (1 - \{p\})}{p},
\end{equation*}
with $\{p\}$ denoting the fractional part of $p$.
\end{theorem}

\begin{theorem}
For any $p \geq 1$, Algorithm~\ref{Pre-Hubs_Algorithm} is an $O \left(\frac{p}{\ln(p + 1)} \cdot \log D \right)$-approximation algorithm for \HLp.
\end{theorem}
\begin{proof}
In order to simplify the analysis and be able to use the above theorem (Theorem~\ref{moment_theorem}) as is, we slightly modify Algorithm~\ref{Pre-Hubs_Algorithm} as follows. After we obtain the set of pre-hubs $\{\widehat{H}_u\}_{u \in V}$ at step 2, we consider the rooted tree $T_u$, defined as $T_u = \bigcup_{u' \in \widehat{H}_u} P_{uu'}$ (observe that this is indeed a tree), for every $u \in V$, and define $F_u \subset V(T_u)$ to be the set of vertices of $T_u$ whose degree (in $T_u$) is at least 3. The modified algorithm then adds the additional step ($2'$): ``$\widehat{H}_u' := \widehat{H}_u \cup F_u$", and then continues the execution of Algorithm~\ref{Pre-Hubs_Algorithm} at step 3, as usual, using the modified sets $\widehat{H}_u'$. Let $ALG$ be the cost of original algorithm, and $ALG'$ be the cost of this modified algorithm. It is not hard to prove that it always holds $ALG \leq ALG'$. It is also easy to see that, since all leaves of $T_u$ are pre-hubs of the set $\widehat{H}_u$, we must have $|F_u| \leq |\widehat{H}_u|$, and so $|\widehat{H}_u'| \leq 2 \cdot |\widehat{H}_u|$. 

Let $\mathcal{P}_u$ be the collection of subpaths of $T_u$ defined as follows: $P$ belongs to $\mathcal{P}_u$ if $P$ is a path between consecutive pre-hubs $u''$ and $u'$ of $\widehat{H}_u'$, with $u''$ being an ancestor of $u'$ in $T_u$, and no other pre-hub $u''' \in \widehat{H}_u$ appears in $P$. For convenience, we exclude the endpoint $u''$ that is closer to $u$: $P = P_{u''u'} - u''$. Note that any such path $P$ is uniquely defined by the pre-hub $u'$ of $u$, and so we will write $P = P_{(uu')}$. The modification we made in the algorithm allows us now to observe that $P \cap P' = \varnothing$, for $P, P' \in \mathcal{P}_u$, $P \neq P'$.

Let $ALG'$ be the cost of the solution $\{H_u\}_u$ that the modified algorithm returns. We have
\begin{equation*}
\begin{split}
    \E[ALG'] &= \E \left[ \left(\sum_{u \in V} |H_u|^p \right)^{1 / p} \right] \leq \left(\sum \E[|H_u|^p] \right)^{1 / p} \quad \quad(\textrm{Jensen's inequality}).
\end{split}
\end{equation*}

We can write $|H_u| \leq \sum_{v \in \widehat{H}_u'} X_v^{u}$, where $X_v^u$ is the random variable indicating how many vertices are added to $H_u$ ``because of" the pre-hub $v$ (see Line 8 of algorithm). Observe that we can write $X_v^u$ as follows: $X_v^u = \sum_{w \in P_{(uv)}} Y_w^{uv}$, with $Y_w^{uv}$ being 1 if $w$ is added in $H_u$, and zero otherwise. The modification that we made in the algorithm implies, as already observed, that any variable $Y_w^{uv}$, $w \in P_{(uv)}$, is independent from $Y_{w'}^{uv'}$, $w' \in P_{(uv')}$, for $v \neq v'$, as the corresponding paths $P_{(uv)}$ and $P_{(uv')}$ are disjoint.

Let $\pi_{uv}: [|P_{(uv)}|] \to P_{(uv)}$ be the permutation we obtain when we restrict $\pi$ (see Line 3 of the algorithm) to the vertices of $P_{(uv)}$. We can then write $\sum_{w \in P_{(uv)}} Y_w^{uv} = \sum_{i = 1}^l Z_i^{uv}$, $l = |P_{(uv)}|$, where $Z_i^{uv}$ is 1 if the $i$-th vertex considered by the algorithm that belongs to $P_{(uv)}$ (i.e. the $i$-th vertex of permutation $\pi_{uv}$) is added to $H_u$ and zero otherwise. It is easy to see that $\Pr[Z_i^{uv} = 1] = 1/ i$. We now need one last observation. We have $\Pr[Z_i^{uv} = 1 \;|\; Z_1^{uv}, ..., Z_{i - 1}^{uv}] = 1/i$. To see this, note that the variables $Z_i^{uv}$ do not reveal which particular vertex is picked from the permutation at each step, but only the relative order of the current draw (i.e. $i$-th random choice) with respect to the current best draw (where best here means the closest vertex to $v$ that we have seen so far, i.e. in positions $\pi_{uv}(1), ..., \pi_{uv}(i -1)$). Thus, regardless of the relative order of $\pi_{uv}(1), ..., \pi_{uv}(i -1)$, there are exactly $i$ possibilities to extend that order when the permutation picks $\pi_{uv}(i)$, each with probability $1/i$. This shows that the variables $\{Z_i^{uv}\}_i$ are independent, and thus all variables $\{Z_i^{uv}\}_{v \in \widehat{H}_v', \;i \in [|P_{(uv)}|]}$ are independent.

We can now apply Theorem \ref{moment_theorem}. This gives
\begin{equation*}
\begin{split}
    \E[|H_u|^p] \leq \E \left[ \left(\sum_{v \in \widehat{H}_u'} \sum_{i  = 1}^{|P_{(uv)}|} Z_i^{uv} \right)^p \right] \leq \left(0.792 \cdot v(p) \cdot \frac{p}{\ln(p + 1)} \right)^p \cdot \Harm_D^p \cdot |\widehat{H}_u'|^p.
\end{split}
\end{equation*}
Here, $\Harm_D = \sum_{i=1}^D \frac{1}{i} = \log D + O(1)$ is the $D$-th harmonic number. Thus,
\begin{equation*}
\begin{split}
    \E[ALG'] &\leq 0.792 \cdot v(p) \cdot \frac{p}{\ln(p + 1)} \cdot \Harm_D \cdot \left(\sum_{u \in V} |\widehat{H}_u'|^p \right)^{1 / p} \\
            &\leq 0.792 \cdot v(p)  \cdot \frac{p}{\ln(p + 1)} \cdot \Harm_D \cdot \left(\sum_{u \in V} 4^p \cdot \left(\sum_{v \in V} x_{uv} \right)^p \right)^{1 / p} \\
            &\leq c \cdot \frac{p}{\ln(p + 1)} \cdot \Harm_D \cdot OPT_{\mathrm{REL}},
\end{split}
\end{equation*}
where $OPT_{\mathrm{REL}}$ is the optimal value of the convex relaxation, and $c \leq 7$ is some constant.

\end{proof}

\section{Trees: Missing proofs of Section~\ref{Trees}}\label{trees_missing_proofs}

\subsection{A PTAS for \HLp on trees}\label{appendix_ptas_lp}

In this section, we describe a polynomial-time approximation scheme (PTAS) for \HLp for arbitrary $p\in [1,\infty)$.
Our algorithm is a modification of the dynamic programming algorithm for \HL1.
The main difficulty that we have to deal with is that the $\ell_p$-cost of an instance cannot be expressed in terms of the $\ell_p$-cost of the subproblems, since 
it might happen that suboptimal solutions for its subproblems give an optimal solution for the instance itself. Thus, it is not enough to store only the cost of the ``optimal" solution for each subproblem.

Let $$OPT[T',t]^p = \min_{H\textrm{ is an HHL for } T'} \sum_{u \in T'} (|H_u| + t)^p.$$
 Clearly, $OPT[T,0]^p$ is the cost of an optimal \HLp solution for $T$, raised to the power $p$. Observe that $OPT[T', t]^p$ satisfies the following recurrence relation:
\begin{equation}\label{eq:OPT-lp-shifted}
OPT[T',t]^p = (1 + t)^p + \min_{r \in T'} \sum_{T'' \textrm{ is a component of }T' - r} OPT[T'', t + 1]^p.
\end{equation}
Indeed, let $\widetilde H$ be an HHL for $T'$ that minimizes $\sum_{u \in T'} (|\widetilde H_u| + t)^p$.
Let $r'$ be the highest ranked vertex in $T'$ w.r.t. the ordering defined by $\widetilde H$. 
For each tree $T''$ in the forest $T'-r'$, consider the hub labeling $\{\widetilde H_u \cap T''\}_{u\in T''} = \widetilde H_u - r'$.
Since $|\widetilde H_u| = |\widetilde H_u \cap T''| +1$, we have
$$\sum_{u \in T''} (|\widetilde H_u| + t)^p = \sum_{u \in T''} (|\widetilde H_u \cap T''| + t+1)^p \geq OPT[T'', t+1]^p.$$
Also, $|\widetilde H_{r'} | =1$. Therefore,
$$OPT[T',t]^p = \sum_{u\in T'} (|\widetilde H_u|+t)^p = (|\widetilde H_{r'}|+t)^p + \sum_{T''}\sum_{u\in T''} (|\widetilde H_u|+t)^p
\geq (1 + t)^p + \sum_{T''} OPT[T'', t + 1]^p.$$
The proof of the inequality in the other direction is similar.
Consider $r$ that minimizes the expression on the right hand side of (\ref{eq:OPT-lp-shifted}) and optimal HHLs
for subtrees $T''$ of $T'-r$. We combine these HHLs and obtain a feasible HHL $\widetilde H$. We get,
$$OPT[T',t]^p \leq (1+t)^p + \sum_{u \in T'} (|\widetilde H_u| + t + 1)^p=
 (1 + t)^p + \sum_{T''} OPT[T'', t + 1]^p.$$
 This concludes the proof of the recurrence.

If we were not concerned about the running time of the algorithm, we could have used
this recursive formula for $OPT[T',t]^p$ to find the exact solution (the running time would be exponential).
 In order to get a polynomial-time algorithm, we again consider only subtrees $T'$ with boundary of size  at most $k$. 
 We consider two cases when $|\partial(T')| < k$ and when $|\partial(T')| = k$. 
 In the former case, we use formula (\ref{eq:OPT-lp-shifted}).
 In the latter case, when $|\partial(T')| = k$, we perform the same step as the one performed in the algorithm for \HL1: we pick a weighted balanced separator vertex $r_0$ of $T'$
 such that $|\partial(T'')| \leq k/2 +1$ for every subtree $T''$ of $T'-r_0$.
 Formally, we define a dynamic-programming table $B[T',t]$ as follows:
\begin{align*}
B[T',t] &= (1 + t)^p + \min_{r \in T'} \sum_{T'' \text{ is a conn. comp. of } T'-r} B[T'',t+1],  &&\text{ if } |\partial(T')| < k,\\
B[T',t] &= (1 + t)^p + \sum_{T'' \text{ is a conn. comp. of } T'-r_0} B[T'',t],  &&\text{ if } |\partial(T')| = k.
\end{align*}
The base cases of our recursive formulas are when the subtree $T'$ is of size 1 or 2. In this case, we simply set $B[T',t] = (1 + t)^p$, if $|T'| = 1$, and $B[T',t] = (1 + t)^p + (2 + t)^p$, if $|T'| = 2$.
We will need the following two claims.
\begin{claim}\label{superadditive}
For any tree $T$ and a partition of $T$ into disjoint subtrees $\{T_1, ..., T_j\}$ such that $\bigcup_{i = 1}^j T_i = T$, we have
\begin{equation*}
\sum_{i = 1}^j OPT[T_i,t]^p \leq OPT[T,t]^p.
\end{equation*}
\end{claim}
\begin{proof}
Consider an optimal hierarchical solution $H$ for the \HLp problem defined by $(T,t)$. Define $H^{(i)} = \{H_u \cap T_i: u \in T_i\}$. Observe that $\{H_u^{(i)}\}_{u \in T_i}$ is a feasible hub labeling for $T_i$, since the original instance is a tree. We have
$$OPT[T,t]^p = \sum_{i = 1}^j \sum_{u \in T_i} (|H_u| + t)^p \geq \sum_{i = 1}^j \sum_{u \in T_i} (|H_u^{(i)}| + t)^p \geq \sum_{i = 1}^j OPT[T_i,t]^p.$$
\end{proof}

\begin{claim}\label{lower_bound_claim_ptas_p_norm}
For any $T'$ and $t \geq 0$, $B[T',t] \leq OPT[T',t]^p$.
\end{claim}
\begin{proof}
We do induction on the size of $T'$. For $|T'| \in \{1,2\}$, the claim holds trivially for all $t \geq 0$. Let us assume that it holds for all subtrees of size at most $s$
and for all $t \geq 0$. We will prove that it holds for all subtrees of size $s+ 1$ and for all $t \geq 0$. We  again consider two cases.

\medskip
\noindent\textbf{Case ${|\partial(T')| < k}$:}
\vspace{-10pt}
\begin{equation*}
\begin{split}
    B[T', t] &= (1 + t)^p + \min_{r \in T'} \sum_{T'' \text{ is comp. of } T'-r} B[T'',t+1] \\
             &\leq (1 + t)^p + \min_{r \in T'} \sum_{T'' \text{ is comp. of } T'-r} OPT[T'', t + 1]^p \quad (\textrm{by ind. hyp.})\\
             &= OPT[T', t]^p.
\end{split}
\end{equation*}

\smallskip
\noindent\textbf{Case $|\partial(T')| = k$:}
\vspace{-10pt}
\begin{equation*}
\begin{split}
    B[T', t] &= (1 + t)^p + \sum_{T'' \text{ is comp. of } T'-r_0} B[T'',t] \\
             &\leq (1 + t)^p + \sum_{T'' \text{ is comp. of } T'- r_0} OPT[T'', t]^p \quad (\textrm{by ind. hyp.})\\
             &= OPT[\{r_0\},t]^p + \sum_{T'' \text{ is comp. of } T'- r_0} OPT[T'', t]^p \\
             &\leq OPT[T', t]^p,
\end{split}
\end{equation*}
where the last inequality follows from Claim \ref{superadditive} and the fact that the connected components of $T'-r_0$ together with $\{r_0\}$ form a partition of $T'$.
\end{proof}

\begin{theorem}
There is a polynomial-time approximation scheme (PTAS) for \HLp for every $p \in [1,\infty)$.
The algorithm finds a $(1+\varepsilon)$ approximate solution in time $n^{O(1/\varepsilon)}$;
the running time does not depend on $p$.
\end{theorem}
\begin{proof}
Fix $\varepsilon < 1$, and set $k = 2 \cdot \lceil 4/\varepsilon \rceil$. Let $H$ be the solution returned by the dynamic programming algorithm
presented in this section. Consider the set $X$ of all weighted balanced separators that the algorithm uses during its execution;
that is, $X$ is the set of hubs $r_0$ that the algorithm adds when it processes trees $T'$ with $|\partial(T')| = k$.
 
Let $\widetilde H_u = (H_u \setminus X) \cup \{u\}$; the set $\widetilde H_u$ consists of the hubs added to $H_u$ during the steps 
when $\partial(T') < k$, with the exception that we include $u$ in $\widetilde H_u$ even if $u\in X$. 
It is easy to prove by induction (along the lines of the previous inductive proofs) that 
$$B[T',t] = \sum_{u \in T'} \left(|\widetilde H_u \cap T'| + t \right)^p.$$
\vspace{-15pt}

\noindent Therefore, $B[T,0] = \sum_{u \in V} |\widetilde H_u |^p$.

Now, consider a vertex $u$ and its hub set $H_u$. We want to estimate the ratio $|H_u\cap X| / |H_u|$. We look at the decomposition tree implied by the algorithm and find the subinstance $T'$ in which the algorithm picked $u$ as the highest ranked vertex in $T'$. The path from the root of the decomposition tree to that particular subinstance $T'$ contains exactly $|H_u|$ nodes. Observe that in any such path, the nodes of the path that correspond to subinstances with boundary size exactly $k$ are at distance at least $k / 2$ from each other (since the size of the boundary increases by at most 1 when we move from one node to the consecutive node along the path). Thus, there can be at most $2|H_u|/k$ such nodes. This means that $|H_u \cap X| \leq 2|H_u| / k$, which gives  $|H_u| \leq (1 + \frac{2}{k - 2}) \cdot |H_u \setminus X|\leq (1 + \frac{2}{k - 2}) \cdot |\widetilde H_u|$. So, the $\ell_p$-cost of the hub labeling is
\begin{equation*}
\begin{split}
    \|H\|_p &= \Bigl(\sum_{u \in V} |H_u|^p \Bigr)^{1/p}\leq \Bigl(1 + \frac{2}{k - 2} \Bigr) \cdot \Bigl(\sum_{u \in V} |\widetilde H_u|^p \Bigr)^{1/p}\\
            &= \Bigl(1 + \frac{2}{k - 2} \Bigr) \cdot B[T,0]^{1/p} \leq \Bigl(1 + \frac{2}{k - 2} \Bigr) \cdot OPT[T,0],
\end{split}
\end{equation*}
where the last inequality follows from Claim \ref{lower_bound_claim_ptas_p_norm}. We get that the algorithm finds a hub labeling of $\ell_p$-cost at most
 $(1 + \frac{2}{k - 2}) \cdot OPT$. The running time is $n^{O(k)} \cdot n = n^{O(k)}$.
\end{proof}

\subsection{A PTAS for \HL\infty on trees}\label{appendix_ptas_infty}

Our approach for \HL1 works almost as is for \HL\infty as well. The only modifications that we need to make are the following:
\begin{itemize}
    \item $B[T']$ is now defined as
\begin{align*}
    B[T'] &= (1 + 4/k) + \min_{r' \in T'} \;\;\max_{T'' \textrm{ is comp. of } T' - r'} B[T''],  &\textrm{for }|\partial(T')| < k, \\
    B[T'] &= \max_{T'' \textrm{ is comp. of } T' - r_0} B[T''],  &\textrm{for }|\partial(T')| = k,
\end{align*}
where $r_0$ is the weighted balanced separator vertex of $T'$, as defined in the description of the algorithm for \HL1.
    \item $C[T']$ is now equal to $C[T'] = \max \left\{ 0, \left(|\partial(T')| - \frac{3k}{4}\right) \cdot \frac{4}{k} \right\}$.
\end{itemize}
We can again prove using induction (along the same lines as the proof for \HL1) that the total cost of the solution that the algorithm returns at any subinstance $T'$ is at most $B[T'] + C[T']$, and that it always holds that $B[T'] \leq (1 + 4/k) \cdot OPT_{T'}$. Thus, for $T' = T$ we have $C[T] = 0$, and so we obtain a solution of cost at most $(1 + 4/k) OPT$, in time $n^{O(k)}$.

\subsection{A quasi-polynomial time algorithm for \texorpdfstring{\HLp}{HLp}}\label{appendix_quasi_lp}

We prove a bound on the infinity norm of an optimal HHL solution for \HLp, and then use that bound as the value of the parameter $k$ in the DP approach of Appendix \ref{appendix_ptas_lp}. We are able to prove a slightly weaker bound (compared to the one obtained for \HL1).
\begin{lemma}\label{p-norm-infinity-upper-bound}
Given a tree $T$, any optimal HHL solution $H$ for \HLp must satisfy $\max_{u \in V} |H_u| = O(\log^2 n)$.
\end{lemma}
\begin{proof}
As already mentioned, the Tree Algorithm (Algorithm~\ref{Tree-Algorithm}) returns a feasible solution $H$ that satisfies $\max_{u \in H} |H_u| = O(\log n)$. Thus, $$\|H\|_p \leq \left(\sum_{u \in V} (c \cdot \log n)^p \right)^{1/p} = c \cdot n^{1 / p} \cdot \log n.$$
This means that there always exists a solution for \HLp of cost at most $c \cdot n^{1 / p} \cdot \log n$.

Since we are again looking at hierarchical solutions, we now need the following observation: every $(2 c \cdot \log n + 1)$ levels at an optimal recursive decomposition, the size of the resulting subproblems must have reduced by at least a factor of 2. To see this, consider the first $2 c \cdot \log n + 1$ levels of the recursion implied by an optimal HHL solution $H$. All subinstances that haven't been solved yet are leaves of the bottom level. Each vertex contained in these subinstances already has $2c \cdot \log n$ hubs. Thus, if we assume that there are strictly more than $\lfloor n / 2 \rfloor$ vertices in these subinstances, then $OPT_T > \frac{n}{2} \cdot 2 c \log n \geq n^{1/p} \cdot c \log n$, which is a contradiction. Thus, we must have at most $\lfloor n / 2 \rfloor$ vertices corresponding to the subproblems that are not yet solved.

Applying this argument inductively to the subproblems that we obtain every $2 c \log n$ consecutive levels of the recursion tree, we get that after $i \cdot 2 c \log n + 1 $ levels of recursion, we will have at most $\lfloor n / 2^i \rfloor$ vertices whose hubs will not have been finalized. It is now clear that no vertices will remain after at most $i = \lceil \log n \rceil + 1$ ``rounds", which implies a total depth of at most $O(\log^2 n)$ levels in the recursion tree. Thus, since the level at which a vertex $u$ is picked as the highest ranked vertex (in the corresponding subinstance in the recursion tree) is exactly equal to the number of hubs that vertex has, we get that $\max_{u \in V} |H_u| = O(\log^2 n)$.
\end{proof}

The lemma implies that by setting $k = O(\log^2 n)$ in the DP approach for \HLp, we can optimally solve \HLp on trees, using the formula
$$ B[T', t] = (1 + t)^p + \min_{r' \in T'} \sum_{T'' \textrm{ is comp. of } T' - r'} B[T'', t + 1],$$
in total time $n^{O(\log^2 n)}$ (with the optimal cost being $B[T,0]^{1/p}$).

\subsection{A quasi-polynomial time algorithm for \texorpdfstring{\HL\infty}{HL-infinity}}\label{appendix_quasi_infinity}

We note that the optimal solution for \HL\infty on a tree has $\ell_\infty$-cost $O(\log n)$,
since the Tree Algorithm (Algorithm~\ref{Tree-Algorithm}) returns a feasible solution $H$ of cost at most $\max_{u \in H} |H_u| = O(\log n)$
(this was first proved by Peleg~\cite{DBLP:journals/jgt/Peleg00}).
Thus, by setting $k = O(\log n)$, we can again use a simple DP algorithm with entries defined by the formula
$$B[T'] = 1 + \min_{r' \in T'} \left(\max_{ T'' \textrm{ is a connected component of } T' - r'} B[T'']\right),$$
and obtain an optimal solution for \HL\infty, in total time $n^{O(\log n)}$.

\subsection{Proof of Lemma \ref{tree-alg-lower-bound}}

\begin{proof}[Proof of Lemma \ref{tree-alg-lower-bound}]\label{proof-of-tree-alg-lower-bound}
We consider the complete binary tree of height $h$, whose size is $n_h = 2^{h + 1} - 1$ (a single vertex is considered to have height 0). The cost of the Tree Algorithm on a complete binary tree of height $h$, denoted by $ALG(h)$, can be written as
\begin{equation*}
ALG(h) =
\begin{cases}
      (2^{h + 1} - 1) + 2 \cdot ALG(h - 1), & h \geq 1, \\
      1 & h = 0.
   \end{cases}
\end{equation*}
It is easy to see that the above implies that $ALG(h) = 2 \cdot h \cdot 2^h + 1$, for all $h \geq 0$. To obtain a $3/2$ gap, we now present an algorithm that gives a hub labeling of size $(1 + o_h(1))
\cdot \frac{4}{3} \cdot h \cdot 2^h$ on complete binary trees (where the $o_h(1)$ term goes to $0$ as $h\to\infty$).

It will be again a recursive algorithm (i.e. a hierarchical labeling), only this time the recursion handles complete binary trees that may have some ``tail" at the root. More formally, the algorithm operates on graphs that can be decomposed into two disjoint parts, a complete binary tree of height $h$, and a path of length $t$. The two components are connected with an edge between the root of the binary tree and an endpoint of the path. Such a graph can be fully described by a pair $(h, t)$, where $h$ is the height of the tree and $t$ is the length of the path attached to the root of the tree.

The proposed algorithm for complete binary trees works as follows. Let $p$ be the root of the tree. Assuming $h \geq 2$, let $l$ be the left child of $p$, and $r$ be the right child of $p$. The algorithm picks $l$ as the vertex with the highest rank, and then recurses on the children of $l$ and on $r$. Observe that on the recursive step for $r$, we have a rooted tree on $r$, and the original root $p$ is now considered part of the tail. For $h = 0$, we have a path of length $t + 1$, and for $h = 1$, we simply remove $p$ and then end up with a path of length $t$ and two single vertices.

Let $\Path(t)$ denote the optimal HL cost for a path of length $t$. It is not hard to show that the Tree Algorithm performs optimally for paths, and a closed formula for $\Path$ is
\begin{equation*}
    \Path(t) = (t + 1) \lceil \log (t + 1) \rceil - 2^{\lceil \log (t + 1) \rceil} + 1, \quad t \geq 0.
\end{equation*}
So, at the base cases, the proposed algorithm uses the Tree Algorithm on paths. Let $P(h, t)$ be the cost of this algorithm. Putting everything together, we obtain the recursive formula
\begin{equation*}
P(h,t) =
\begin{cases}
      (2^{h + 1} - 1) + t + 2 \cdot P(h - 2, 0) + P(h - 1, t + 1), & h \geq 2, \;t \geq 0, \\
      5 + t + \Path(t) & h = 1, \;t \geq 0, \\
      \Path(t + 1), & h = 0, \;t \geq 0.
   \end{cases}
\end{equation*}

The cost of the solution we obtain from this algorithm is $P(h, 0)$. Let $f(n) = \Path(n + 1) + 5$, $n \geq 0$, and $g(h) = C / \sqrt{h}$, $h \geq 0$, for some appropriate constant $C$. We prove by induction on $h$ that
$$
P(h, t) \leq \frac{4}{3} \cdot h \cdot 2^h + g(h) \cdot h \cdot 2^h + f(h + t) + h \cdot t, \quad\forall h \geq 0, t \geq 0.
$$

The cases with $h = 0$ and $h = 1$ are obvious. If $h \geq 2$, then
$$
P(h, t) = 2^{h + 1} - 1 + t + 2 \cdot P(h - 2, 0) + P(h - 1, t + 1).
$$
By the induction hypothesis, we have that
\begin{align*}
2 \cdot P(h - 2, 0) & \leq \frac{1}{2} \cdot \frac{4}{3} h \cdot 2^h - \frac{4}{3} 2^h + \frac{1}{2} g(h - 2) \cdot (h - 2) 2^h + 2 \cdot f(h - 2), \quad \textrm{and}\\
P(h - 1, t + 1) & \leq \frac{1}{2} \cdot \frac{4}{3} h \cdot 2^h - \frac{2}{3} 2^h + \frac{1}{2} g(h-1) \cdot (h - 1) 2^h + f(h + t) + (h - 1)\cdot(t + 1).
\end{align*}

Thus, we obtain
$$ P(h, t) \leq \frac{4}{3} h \cdot 2^h + \left(\frac{g(h-2) \cdot(h-2) + g(h - 1) \cdot (h - 1)}{2} + \frac{h + 2f(h-2) - 2}{2^h} \right) \cdot 2^h + f(h + t) + h \cdot t.$$

Choosing the right constant $C$, we can show that for all $h \geq 2$, we have
$$
\frac{g(h-2) \cdot(h-2) + g(h - 1) \cdot (h - 1)}{2} + \frac{h + 2f(h-2) - 2}{2^h} \leq h \cdot g(h),
$$
and so the inductive step is true. This means that
$$ P(h,0) \leq \frac{4}{3} \cdot h \cdot 2^h \cdot \left(1 + \frac{3}{4} g(h) + \frac{3 f(h)}{4h \cdot 2^h} \right) = \left(1 + o_h(1)\right) \cdot  \frac{4}{3} \cdot h \cdot 2^h,$$
(where the $o_h(1)$ term goes to $0$ as $h\to\infty$)
and so, for any $\varepsilon > 0$ there are instances where $\frac{ALG}{OPT} \geq \frac{3}{2} - \varepsilon$.
\end{proof}

\section{Any ``natural" rounding scheme cannot break the $O(\log n)$ barrier for $HL_1$ on graphs with unique shortest paths and diameter $D$}
\label{lower_bound_rounding}

In this section, we show that any rounding scheme that may assign $v \in H_u$ only if $x_{uv} > 0$ gives $\Omega(\log n)$ approximation, even on graphs with shortest-path diameter $D = O(\log n)$. For that, consider the following tree $T$, which consists of a path $P = \{1, ..., k\}$ of length $k = 3t$, $t \in \mathbb{N} \setminus \{0\}$, and two stars $\mathcal{A}$ and $\mathcal{B}$, with $N = \binom{k}{2t}$ leaves each (each leaf corresponding to a subset of $[k]$ of size exactly $2t$). The center $a$ of $\mathcal{A}$ is connected to vertex ``1" of $P$ and the center $b$ of $\mathcal{B}$ is connected to vertex ``$k$" of $P$. The total number of vertices of $T$ is $n = 2N + 2 + k$, which implies that $t = \Omega(\log n)$.

\begin{figure}[h]
\begin{center}
\scalebox{1.1}{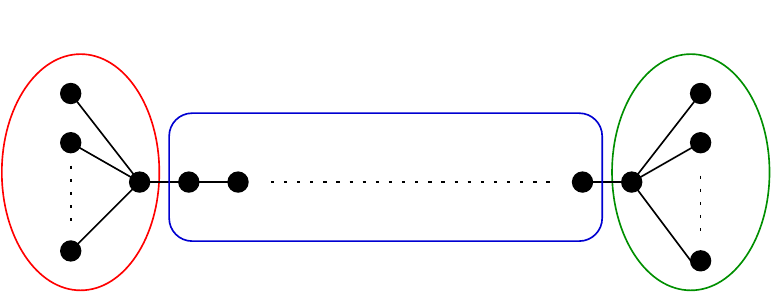}
\caption{An instance that cannot be rounded well with any ``natural" rounding scheme}
\end{center}
\end{figure}

Consider now the following LP solution (all variables not assigned below are set to zero):
\begin{itemize}
    \item $x_{uu} = 1$, for all $u \in T$.
    \item $x_{Sa} = 1$, for all $S \in \mathcal{A}$.
    \item $x_{Wb} = 1$, for all $W \in \mathcal{B}$.
    \item $x_{Si} = 1/t$, for all $S \in \mathcal{A}$, $i \in S \subseteq P$.
    \item $x_{Wi} = 1/t$, for all $W \in \mathcal{B}$, $i \in W \subseteq P$.
    \item $x_{ab} = x_{ba} = 1$.
    \item $x_{ia} = x_{ib} = 1$, for all $i \in [k]$.
    \item $\{x_{ij}\}_{i,j \in [k]}$ is an optimal solution for $P$.
\end{itemize}
Observe that the above solution is indeed a feasible fractional solution. Its cost is at most $n + 3(|\mathcal{A}| + |\mathcal{B}|) + 2 + 2k + c \cdot k \cdot \log k = \Theta(n)$, for some constant $c$. Suppose now that we are looking for a rounding scheme that assigns $v \in H_u$ only if $x_{uv} > 0$, and let's assume that there exists a vertex $S \in \mathcal{A}$ whose resulting hub set satisfies $|H_S \cap P| < t$. We must also have $H_S \cap \mathcal{B} = \varnothing$, since $x_{Su} = 0$ for all $u \in \mathcal{B}$. This implies that there exists a $W \in \mathcal{B}$ such that $W \cap H_S = \varnothing$. Since the above fractional solution assigns non-zero values only to $x_{Wi}$ with $i \in W$ and $x_{Wb}$, this means that $x_{Wi} = 0$ for all $i \in H_S$. Thus, the resulting hub set cannot be feasible, which implies that any rounding that satisfies the aforementioned property and returns a feasible solution must satisfy $|H_S \cap P| \geq t$ for all $S \in \mathcal{A}$ (similarly, the same holds for all $W \in \mathcal{B}$). This means that the returned solution has cost $\Omega(n \cdot t) = \Omega(n \cdot \log n)$, and so the approximation factor must be at least $\Omega(\log n)$.

\section{Hub labeling on directed graphs}\label{appendix_directed}

In this section, we sketch how some of the presented techniques can be used for the case of directed graphs. Let $G = (V, E)$ be a directed graph with edge lengths $l(e) > 0$. Instead of having one set of hubs, each vertex $u$ has two sets of hubs, the forward hubs $H_u^{(f)}$ and the backward hubs $H_u^{(b)}$. The covering property is now stated as follows: for every (directed) pair $(u,v)$ and some directed shortest path $P$ from $u$ to $v$, we must have $H_u^{(f)} \cap H_v^{(b)} \cap P \neq \varnothing$. The \HLp objective function can be written as $\left(\sum_{u \in V} L_u^p \right)^{1 / p}$, where $L_u = |H_u^{(f)}| + |H_u^{(b)}|$.

The Set Cover based approach of Cohen et al. (\cite{DBLP:journals/siamcomp/CohenHKZ03}) and Babenko et al. (\cite{DBLP:conf/icalp/BabenkoGGN13}) can be used in this setting in order to obtain an $O(\log n)$ approximation for \HLp, $p \in [1, \infty]$. It is also straightforward to see that there is a very simple 2-approximation preserving reduction from undirected \HLp to directed \HLp, implying that an $\alpha$-approximation for directed \HLp would give a $2\alpha$-approximation for undirected \HLp. Thus, the hardness results of Section~\ref{Hardness} can be applied to the directed case as well, and so we end up with the following theorem.
\begin{theorem}
\HLp is $\Omega(\log n)$-hard to approximate in directed graphs with multiple shortest paths, for $p \in \{1, \infty\}$, unless \P = \NP.
\end{theorem}

Having matching lower and upper bounds (up to constant factors) for the general case, we turn again to graphs with unique (directed) shortest paths. The notion of pre-hubs can be extended to the directed case as follows: a family of sets $\{(\widehat{H}_u^{(f)}, \widehat{H}_u^{(b)})\}_{u \in V}$ is a family of pre-hubs if for every pair $(u,v)$ there exist $u' \in \widehat{H}_u^{(f)} \cap P_{uv}$ and $v' \in \widehat{H}_v^{(b)} \cap P_{uv}$ such that $u' \in P_{v'v}$.

We now present the LP relaxation for the $\ell_1$ case.
\\ \\
\noindent($\mathbf{DIR-LP_1}$)
\begin{equation*}
\begin{split}
    \min: &\quad \sum_{u \in V} \sum_{v \in V} \left(x_{uv}^{(f)} + x_{uv}^{(b)}\right) \\
    \text{s.t.:} &\quad \sum_{w \in P_{uv}} \min\{x_{uw}^{(f)}, x_{vw}^{(b)}\} \geq 1,  \quad \forall (u, v) \in V \times V, \\
          &\quad x_{uv}^{(f)} \geq 0,  \hspace{93pt} \forall (u,v) \in V \times V, \\
          &\quad x_{uv}^{(b)} \geq 0,  \hspace{95pt} \forall (u,v) \in V \times V.
\end{split}
\end{equation*}

In order to obtain a feasible set of pre-hubs, for each vertex $u \in V$ we construct two trees (both rooted at $u$): $T_u^{(f)}$ is the union of all directed paths from $u$ to all other vertices, and $T_u^{(b)}$ is the union of all directed paths to $u$ from all other vertices. We drop the orientation on the edges, and we note that these are indeed trees. We proceed as in the undirected case and obtain a set $\widehat{H}_u^{(f)}$ from $T_u^{(f)}$ (using the variables $x_{uv}^{(f)}$) of size at most $2 \sum_{v \in V} x_{uv}^{(f)}$, and a set $\widehat{H}_u^{(b)}$ from $T_u^{(b)}$ (using the variables $x_{uv}^{(b)}$) of size at most $2 \sum_{v \in V} x_{uv}^{(b)}$. It is not hard to see that the obtained sets $(\widehat{H}_u^{(f)}, \widehat{H}_u^{(b)})_{u \in V}$ are indeed pre-hubs.

We can now use a modified version of Algorithm \ref{Pre-Hubs_Algorithm}, as given below.
\\[4pt]
\begin{algorithm}[H]
    Solve $\mathbf{DIR-LP_1}$ and get an optimal solution $\{(x_{uv}^{(f)}, x_{uv}^{(b)})\}_{(u,v) \in V \times V}$.\\
    Obtain a set of pre-hubs $\{(\widehat{H}_u^{(f)}, \widehat{H}_u^{(b)})\}_{u \in V}$ from $x$. \\
    Generate a random permutation $\pi : [n] \to V$ of the vertices. \\
    Set $(H_u^{(f)}, H_u^{(b)}) = (\varnothing, \varnothing)$, for every $u \in V$. \\
    \For{$i= 1$ \KwTo $n$} {
        \For{every $u \in V$} {
            \For{every $u' \in \widehat{H}_u^{(f)}$, such that $\pi_i \in P_{uu'}$ and $P_{\pi_i u'} \cap \widehat{H}_u^{(f)} = \{u'\}$} {
                \lIf{$P_{\pi_i u'} \cap H_u^{(f)} = \varnothing$}{$H_u^{(f)} := H_u^{(f)} \cup \{\pi_i\}$}
            }
            \For{every $u' \in \widehat{H}_u^{(b)}$, such that $\pi_i \in P_{u'u}$ and $P_{u' \pi_i} \cap \widehat{H}_u^{(b)} = \{u'\}$} {
                \lIf{$P_{u' \pi_i} \cap H_u^{(b)} = \varnothing$}{$H_u^{(b)} := H_u^{(b)} \cup \{\pi_i\}$}
            }
        }
    }
    Return $\{(H_u^{(f)}, H_u^{(b)})\}_{u \in V}$.
    \bigskip

\caption{Algorithm for \HL1 on directed graphs with unique shortest paths}
\label{DIR-Pre-Hubs_Algorithm}
\end{algorithm}
\vspace{10pt}
It is easy to see that the obtained solution is always feasible, and, with similar analysis as before, we prove that $\E[|\widehat{H}_u^{(f)}|] \leq 2(\log D + O(1)) \cdot \sum_{v} x_{uv}^{(f)}$ and $\E[|\widehat{H}_u^{(b)}|] \leq 2(\log D + O(1)) \cdot \sum_{v} x_{uv}^{(b)}$. Thus, in expectation, we obtain a solution of cost $O(\log D) \cdot OPT_{DIR-LP_1}$.

The analysis can also be generalized for arbitrary fixed $p \geq 1$, similar to the analysis of Appendix \ref{sec:lp-norm-algorithm}. The algorithm is modified in the same way, and using the fact that for $x,y,p \geq 1$, we have $x^p + y^p \leq (x + y)^p \leq 2^p (x^p + y^p)$, we can again obtain an approximation of the form $c \cdot \frac{p}{\ln (p + 1)} \Harm_D \cdot OPT_{REL}$, where $OPT_{REL}$ is the optimal value of the corresponding convex relaxation. Thus, we obtain the following theorem.
\begin{theorem}
There is an $O\left( \frac{p}{\ln(p+1)} \cdot \log D \right)$-approximation algorithm for \HLp, $p \geq 1$, on directed graphs with unique shortest paths.
\end{theorem}

\end{document}